\documentclass[pra,twocolumn,showpacs,groupedaddress,notitlepage,superscriptaddress,floatfix,nofootinbib]{revtex4-2}

\usepackage{subfigure}
\usepackage{braket}
\usepackage{tikz}
\usepackage{amsmath}
\usepackage{listings}
\usepackage{algorithmicx}
\usepackage[ruled]{algorithm} 
\usepackage{algpseudocode}
\usepackage{amsfonts}
\usepackage{float,mathdots}
\usepackage{amssymb}
\usepackage{accents}
\usepackage{amsthm}
\usepackage{hyperref}
\usepackage{blkarray}
\usepackage{svg}
\usepackage{booktabs}

\newcommand{\ceil}[1]{\left\lceil #1 \right\rceil}

\newtheorem{theorem}{Theorem}[]
\newtheorem{lemma}{Lemma}[]

\newtheorem{proposition}{Proposition}[]

\newtheorem{definition}{Definition}[]

\newtheorem{example}{Example}[]

\renewcommand{\[}{\begin{equation}}
\renewcommand{\]}{\end{equation}}

\usepackage{bbm}
\usepackage{appendix}

\hypersetup{
    colorlinks=true,   
    linkcolor=cyan,    
    citecolor=magenta, 
    filecolor=magenta, 
    urlcolor=cyan,     
    runcolor=cyan
}
\usepackage[capitalise]{cleveref} 

\renewcommand\vec{\mathbf}

\newcommand{\cyan}[1]{\textcolor{cyan}{#1}}

\newcommand{\DL}[1]{\cyan{[DL: #1]}}
\renewcommand{\DL}[1]{}

\begin{document}


\title{Continuous measurement-based holonomic quantum computation}

\author{Anirudh Lanka}
\affiliation{Department of Electrical and Computer Engineering, University of Southern California, Los Angeles, California}

\author{Juan Garcia-Nila}
\affiliation{Department of Electrical and Computer Engineering, University of Southern California, Los Angeles, California}

\author{Todd A. Brun}
\affiliation{Department of Electrical and Computer Engineering, University of Southern California, Los Angeles, California}

\begin{abstract}
We propose a scheme to generate holonomies using the Quantum Zeno effect, enabling logical unitary operations on quantum stabilizer codes purely through measurements. The quantum error-correcting code space is adiabatically rotated by measuring a succession of rotated stabilizer generators. When the rotation is sufficiently slow, the state remains confined to the instantaneous code space by the Zeno effect; otherwise, measurement-induced jumps can occur into a rotated orthogonal subspace. If the rotation completes a closed loop, the code state is transformed by a holonomy: a logical unitary transformation. We analytically derive the sequence of rotated stabilizer generators that produce a desired holonomy, and find the total time required to implement this procedure with a given success probability. If a measurement moves the state to the orthogonal subspace, we present a method to alter the path of the rotated observables to return the state either to the original code or the original error space with the desired holonomy; in the latter case, the holonomy is emulated. Finally, we establish conditions on the code and the measured observables that preserve the correctability of a given error set. When a code fails to meet the error-correcting conditions, our protocol can be applicable by augmenting the code with ancilla qubits.

\end{abstract}

\maketitle

\section{Introduction}

Achieving robust universal quantum computation requires a delicate balance between isolating computational degrees of freedom, to minimize environmental decoherence, and maintaining precise dynamical control. Limitations due to this trade-off often result in noise-induced errors. To mitigate them, quantum information is encoded in higher-dimensional systems, with computation performed on the encoded qubits \cite{Shor_QEC, Steane_QEC, Lidar_Brun_2013}. In circuit-based quantum computing (CQC), universal logical gate sets typically involve multi-qubit interactions, which can propagate errors and make them uncorrectable, requiring fault-tolerant design \cite{resilient_klz, Shor_FTQC, kitaev_2003}. While transversal gates are fault-tolerant, the Eastin-Knill theorem limits their universality \cite{eastin_knill_PRL}.

Geometric approaches, such as holonomic and topological quantum computation, offer promising pathways to universal fault-tolerant computation \cite{zanardi_1999, kitaev_2003, freedman2002topologicalquantumcomputation, sjoqvist_non-adiabatic_2012, Lidar_Brun_2013, Lahtinen_2017, Lidar_holonomyDFS, zheng_brun, oreshkov_subsystem_hqc, zhang_surface_hqc, hqc_stab_codes}. Holonomic Quantum Computing (HQC) performs logical operations by using the geometric phase acquired by traversing a parameterized code space along loops in a suitable manifold, where the path determines the resulting unitary operation \cite{zanardi_1999}. This method has been shown to be robust against certain systematic errors \cite{Cesare_2015, oreshkov_2009, buividovich_fidelity_2006, solinas_robustness_2004, chen_robust_2020, liang_robust_2025, shen_accelerated_2023, florio_robust_2006, solinas_stability_2012, fuentes_guridi_holonomic_2005, sarandy_abelian_2006}. Conversely, Topological Quantum Computing (TQC) relies on the braiding of anyonic excitations, and is inherently fault-tolerant due to its dependence on the topology rather than the specific geometry of the path \cite{kitaev_2003, non_abelian_tqc}. While TQC remains mostly theoretical, experimental demonstrations of HQC across various quantum architectures have been proposed and some achieved \cite{wassner_holonomic_2025, feng_experiment_nonad, toyoda_experim_holon, Solinas_semicond, albert_hqc_cat, trapped_ion_hqc}.

Conventionally, HQC is performed by encoding quantum information in an eigenspace of a degenerate Hamiltonian and adiabatically evolving the system along closed loops in a control parameter space \cite{sjoqvist_non-adiabatic_2012, Zheng_2015, zheng_brun, farhi2000}. This adiabatic evolution preserves the degeneracy of the Hamiltonian and induces non-Abelian geometric phases---holonomies---that enable universal quantum computation. However, adiabatic evolution is not the only mechanism by which degeneracy can be preserved. It is well-established that the adiabatic theorem and the quantum Zeno effect are closely related \cite{facchi_qzs, facchi_2003, kitano_zeno, burgarth_onebound, facchi_comparison_2005}. In this paper, we explore the possibility of generating arbitrary holonomies through frequent measurements. A similar approach, using sequences of projective measurements to implement \emph{discrete} holonomies, was explored in \cite{Mommers_2022} for spin-coherent states (SCS). In this work, we provide an in-depth analysis of both discrete- and continuous-path holonomies for arbitrary stabilizer codes.

Discrete holonomies can be generated by incrementally rotating the quantum error-correcting code space and measuring the corresponding rotated projector, continuing the process until the rotation completes a closed loop. The resulting holonomy corresponds to a logical unitary gate that depends on the chosen path. When the rotation increment is sufficiently small, the state is projected to the rotated code space with high probability due to the Zeno effect. In the limit of infinitely many \emph{weak} measurements and infinitesimal rotations, we can generate continuous-path holonomies. Interestingly, we find that the probability to stay in the code space differs between the discrete and continuous approaches.

The measurement-based holonomic approach is particularly appealing since it also provides partial information about the state of the system. This is in general sufficient to identify whether the system is in the correct eigenspace without further disturbing the state. If it is not, one can appropriately control the system to bring it back to the desired eigenspace. This is true regardless of whether the rotations are discrete or continuous. This is unlike conventional Hamiltonian-based HQC, where there are no measurements to reveal errors along the path.

A quantum error-correcting code is associated with a (non-unique) ``correctable'' set of errors, denoted by $\mathcal{E}$. The Knill-Laflamme condition gives necessary and sufficient criteria for a code to correct $\mathcal{E}$. During the holonomic procedure, as the code space is rotated, the instantaneous codes should also be able to correct the same error set $\mathcal{E}$. We find sufficient conditions to satisfy this requirement for arbitrary stabilizer codes. For certain codes that do not satisfy these conditions, we can make the procedure work by appending ancilla qubits.

The main results of this work are:
\begin{itemize}
    \item We introduce a framework to generate a continuous-path holonomy using a sequence of both projective and weak measurements.
    \item We analytically derive the probabilities of staying in the desired eigenspace in both cases.
    \item In case of a measurement-induced error, where the system is projected into a state orthogonal to the instantaneous code space, we show a way to dynamically modulate the path to correct it, and bring the state back either to the original code or the error space.
    \item We find sufficient conditions for the instantaneous code to satisfy the error-correcting conditions throughout the loop.
\end{itemize}

This paper is organized as follows: we provide preliminary information about the dynamics in quantum Zeno subspaces and find the associated holonomy in \cref{sec:prelims}. We construct a framework to perform measurement-based HQC in \cref{sec:mhqc}. We first consider discrete holonomy in \cref{sec:proj_meas}, and show how to achieve implicit fault-tolerance, in \cref{sec:fault_tolerance}. In \cref{sec:continuous_meas}, we extend that framework to allow continuous weak measurements while simultaneously rotating the code space. We derive the error-correcting conditions and show how to satisfy them throughout the evolution in \cref{sec:qec_conds}. We give concluding remarks and discuss the future scope of this work in \cref{sec:discussion}.

\section{Preliminaries}\label{sec:prelims}

\subsection{Quantum Zeno dynamics}\label{sec:qze}

Suppose a quantum system with Hilbert space $\mathcal{H}$ undergoes a series of measurements of a time-dependent family of observables
\[
    O(t)=U(t, 0)O_0U^\dagger(t, 0),
\]
where $U(t, 0)$ is a family of unitary matrices, and $O_0 \equiv O(0)$. For the measurements to have a nontrivial effect on the state, we require that
\[
    [U(t_1, 0), U(t_2, 0)] \neq 0 \quad \forall t_1 \neq t_2 .
\]
In a small time interval $\delta t$, the evolution unitary is
\[
    U(t+\delta t, t) = \exp(-iH(t)\delta t) + \mathcal{O}(\delta t^2),
\]
where $H(t)$ is a Hamiltonian that can be time-dependent. When $\delta t \ll 1$, then $U(t+\delta t, t) \approx I$, i.e., the unitary is \emph{weak}; the observable $O(t)+\delta t$ is rotated by a small amount from $O(t)$. Observe that $O(t)$ always remains Hermitian, and its eigenvalues are unchanged, since unitary transformations preserve eigenvalues. We can write the rotated observable as
\[
    O(t)=\sum_s \alpha_s \mathbb{P}_s(t),
\]
where $\mathbb{P}_s(t)$ is the projector onto the $\alpha_s$-eigenspace; let us denote this eigenspace by $\mathcal{H}_s(t)$. We can also write $O(t)$ as
\[
O(t) = \sum_s \alpha_s \left[U(t, 0) \mathbb{P}_s(0) U^\dagger(t, 0)\right],
\]
i.e., the projectors onto the eigenspaces of the rotated observables also evolve through the same unitary transformation.

Suppose that a state $\ket{\psi(t)} \in \mathcal{H}_0(t)$ undergoes a projective measurement of $O(t+\delta t)$. The probability that the state will be left in $\mathcal{H}_0(t+\delta t)$ is
\[
\begin{aligned}
p_0 &= \bra{\psi(t)} \mathbb{P}_0(t+\delta t) \ket{\psi(t)} \\
    &= \bra{\psi(t)} U(t+\delta t, t) \mathbb{P}_0(t) U^\dagger(t+\delta t, t) \ket{\psi(t)} \\
    &= 1 - \mathcal{O}(\delta t^2),
\end{aligned}
\]
i.e., starting from the $\alpha_0$-eigenspace, when the time between the subsequent measurements $\delta t \ll 1$, the probability of being left in an eigenspace with the same eigenvalue is close to unity. In other words, all states are ``confined'' to their respective instantaneous subspaces throughout the evolution; these subspaces are formally known as \emph{Quantum Zeno Subspaces} (QZS). The state evolves into
\[
\begin{split}
    \ket{\psi(t+\delta t)} &= \frac{\mathbb{P}_0(t+\delta t)\ket{\psi(t)}}{\sqrt{p_0}} \\
        &= \left(\mathbb{P}_0(t) -i[H(t), \mathbb{P}_0(t)]\delta t \right) \ket{\psi(t)} + \mathcal{O}(\delta t^2).
\end{split}
\]
In the limit $\delta t \rightarrow 0$, we obtain a Markovian differential equation for the state evolution:
\[\label{eq:zeno_dynamics}
    \ket{\mathrm{d}\psi} = -i[H(t), \mathbb{P}_0(t)]\ket{\psi}\mathrm{d}t.
\]
It is important to note that this evolution is \emph{not} unitary since $\tilde{H}(t) \equiv [H(t), \mathbb{P}_0(t)]$ is anti-Hermitian. However, as we shall see, such dynamics can produce a closed-loop holonomy within a Zeno subspace.

\subsection{Generating holonomies}

We now introduce several algebraic constructions that are useful in finding the holonomy. We follow the theory presented in \cite{zheng_brun}, and show the equivalence between the obtained holonomies in the adiabatic and Zeno regimes. This gives yet more evidence that the two effects are analogous.

Suppose that we are given a collection of independent vector subspaces. A natural geometric description of this collection is the complex Grassmannian manifold $\mathcal{B} \equiv G_{N,K}(\mathbb{C})$ defined by
\[
    \mathcal{B} = \{\mathbb{P} \in M(N, N; \mathbb{C})| \mathbb{P}^2=\mathbb{P}, \mathbb{P}^\dagger=\mathbb{P}, \text{Tr}\{\mathbb{P}\}=K\},
\]
where $M(N, N; \mathbb{C})$ is the set of all $N \times N$ complex matrices. Hence, each point on $G_{N, K}(\mathbb{C})$ is a projector onto a $K$-dimensional subspace of $\mathbb{C}^N$. The set of basis vectors for a given subspace is not unique; different choices of basis are related by transformations from the unitary group $U(K)$. The \emph{fiber} at a given point in $\mathcal{B}$ consists of all possible orthonormal bases for that subspace. The collection of all such fibers forms a \emph{fiber bundle}, known as the complex Stiefel manifold $\mathcal{F}\equiv S_{N, K}(\mathbb{C})$, which is defined as
\[
    \mathcal{F} = \{L \in M(N, K; \mathbb{C})|L^\dagger L=I_K\},
\]
where $I_K$ is the $K$-dimensional identity. Intuitively, the columns of $L$ form an orthonormal basis for a subspace. In the context of quantum error-correcting codes, a point $\mathbb{P} \in \mathcal{B}$ represents the projector onto the code space, while a point $L$ in the fiber over $\mathbb{P}$ contains the code basis states.

The left action of an $N$-dimensional unitary matrix gives a transformation within $G_{N, K}(\mathbb{C})$:
\[
    U(N) \times \mathcal{B} \rightarrow \mathcal{B}, \;\;\;\; (g, \mathbb{P}) \mapsto g\mathbb{P}g^\dagger.
\]
Define the projection map $\pi: \mathcal{F} \rightarrow \mathcal{B}$ as
\[
    \pi(L) \equiv LL^\dagger.
\]
The right action of a $K\times K$ unitary matrix $h$ gives a transformation within the fiber associated with a point in $\mathcal{B}$:
\[\label{eq:transform_fiber}
    \mathcal{F} \times U(K) \rightarrow \mathcal{F}, \;\;\;\; (L,h) \mapsto Lh.
\]
It is easy to see that the bases $V$ and $Vh$ span the same subspace and hence correspond to the same point on $G_{N, K}(\mathbb{C})$, since
\[
    \pi(Lh) = (Lh)(Lh)^\dagger = Lhh^\dagger L^\dagger = LL^\dagger = \pi(L).
\]
Finally, the tuple $\left(\mathcal{F}, \mathcal{B}, \pi, U(K)\right)$ that we operate on is called the \emph{principal bundle}.

Suppose that a point $\mathbb{P}_0$ on $\mathcal{B}$ moves along a curve. Without loss of generality, assume that this motion arises by rotating the subspace projector $\mathbb{P}_0$ with a parameterized unitary operator $V(t)$:
\[
    \mathbb{P}(t) = V(t) \mathbb{P}_0 V^\dagger(t).
\]
The associated curve on $\mathcal{F}$ depends on its \emph{connection}. The canonical connection on $\mathcal{F}$ is defined as a $\mathfrak{u}(K)$-valued $1$-form on $\mathcal{B}$:
\[\label{eq:connection}
    A=L^\dagger(t) \frac{\mathrm{d}L(t)}{\mathrm{d}t}.
\]
This is the unique connection that is invariant under the transformation in \cref{eq:transform_fiber}:
\[
\begin{aligned}
    \tilde{A} &= (L(t)h)^\dagger \frac{\mathrm{d}(L(t)h)}{\mathrm{d}t} \\
    &= h^\dagger Ah + h^\dagger\mathrm{d}h,
\end{aligned}
\]
which is called a \emph{gauge} transformation.

We are now ready to apply this formalism to measurement-based holonomic quantum computation. For each $t$, there exists $V(t) \in \mathcal{F}$ such that $\mathbb{P}(t) = L(t)L^\dagger(t)$. In the limit of frequent measurements by $\mathbb{P}(t)$, the Zeno effect maintains the state in the instantaneous code space throughout the evolution by the argument presented in \cref{sec:qze}, and hence we can substitute the state $\ket{\psi(t)} \in \mathbb{C}^N$ with the reduced state vector $\ket{\phi(t)} \in \mathbb{C}^K$:
\[\label{eq:reduced_sv}
    \ket{\psi(t)} = L(t)\ket{\phi(t)}.
\]
Substituting \cref{eq:reduced_sv} in \cref{eq:zeno_dynamics}, we obtain
\[
    \frac{\mathrm{d}\big(L(t)\ket{\phi(t)}\big)}{\mathrm{d}t} = -i\big[H(t), \mathbb{P}(t) \big] L(t)\ket{\phi(t)}.
\]
Using $L(t)L^\dagger(t)=\mathbb{P}(t)$ and $L^\dagger(t)L(t) = I_K$, we have
\[
    \ket{\mathrm{d}\phi} + L^\dagger\mathrm{d}L\ket{\phi(t)} = 0,
\]
whose solution can be represented as
\[
    \ket{\phi(t)} = \mathcal{P}\exp\left(-\int L^\dagger \mathrm{d}L\right)\ket{\phi(0)},
\]
where $\mathcal{P}$ indicates path ordering. Therefore, $\ket{\psi(t)}$ can be written as
\[\label{eq:final_state_with_vertical_lift}
    \ket{\psi(t)} = L(t)\mathcal{P}\exp\left(-\int L^\dagger \mathrm{d}L\right)L^\dagger(0)\ket{\psi(0)}.
\]
If the subspace evolves through a loop and returns to the initial subspace, $\mathbb{P}(T) = \mathbb{P}(0)$, the holonomy $\Gamma \in U(K)$ is defined as
\[
    \Gamma = L^\dagger(0)L(T)\mathcal{P}\exp\left(-\int L^\dagger \mathrm{d}L\right),
\]
and the final state is
\[
    \ket{\psi(T)} = L(0)\Gamma\ket{\phi(0)}.
\]

A general solution for the curve $L(t)$ is
\[\label{eq:general_stiefel_curve}
    L(t) = V(t)L(0)h(t),
\]
where $h(t) \in U(K)$. Intuitively, this can be thought of as a two-step process: 1) rotate the basis states $L(0)$ using $V(t)$ so that $\pi(\tilde{L}(t)) = \mathbb{P}(t)$, where $\tilde{L}(t) \equiv V(t)L(0)$; 2) rotate the basis states within the fiber using $h(t)$, giving $L(t)$. If the condition
\[\label{eq:horizontal_lift}
    L^\dagger \frac{\mathrm{d}L}{\mathrm{d}t} = 0 \; \,\forall t
\]
is satisfied, the curve $L(t) \in \mathcal{F}$ is called the {\it horizontal lift} of the curve $\pi(L(t)) \in \mathcal{B}$ (as shown in \cref{fig:holonomy}). Then the holonomy simplifies to
\[
    \Gamma = L^\dagger(0) L(T) \in U(K) .
\]

\begin{figure}
    \centering
    \includegraphics[width=0.8\linewidth]{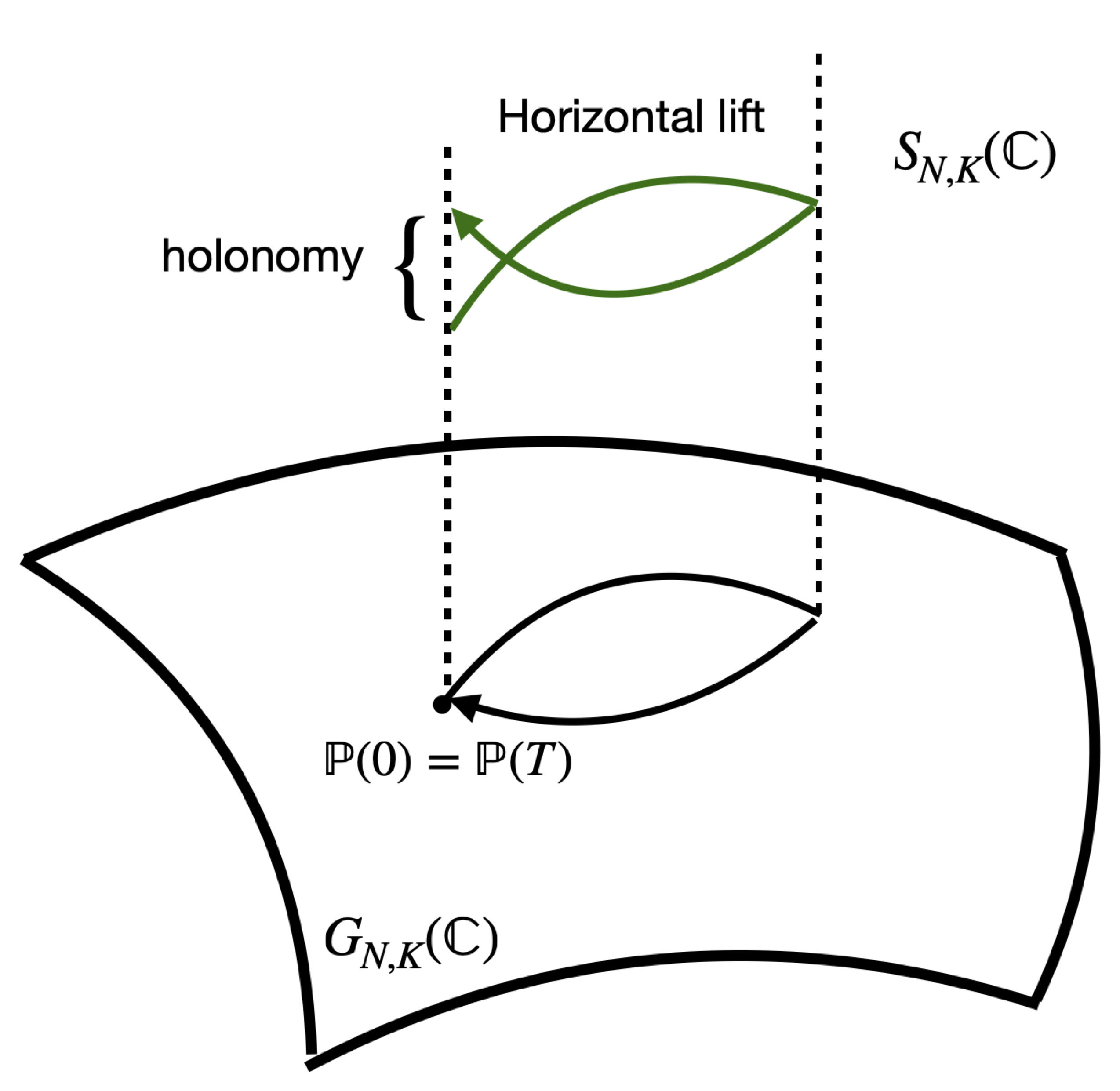}
    \caption{The horizontal lift as a unique curve in $S_{N, K}(\mathbb{C})$ with the base manifold $G_{N, K}(\mathbb{C})$. The difference between the initial point $V(0)$ and the final point $V(T)$ is the holonomy.}
    \label{fig:holonomy}
\end{figure}

Substituting \cref{eq:general_stiefel_curve} into \cref{eq:horizontal_lift}, we obtain a differential equation for $h(t)$:
\[\label{eq:lift_deq}
    \frac{\mathrm{d}h(t)}{\mathrm{d}t} = -L^\dagger(0)V^\dagger(t)\frac{\mathrm{d}V(t)}{\mathrm{d}t}L(0)h(t).
\]
It is a well-known result of differential geometry that the horizontal lift associated with a given connection for a principal bundle is unique.
\begin{proposition}
    Suppose $\gamma:[0, 1] \rightarrow \mathcal{B}$ is a curve in $\mathcal{B}$, and let $L_0 \in \pi^{-1}(\mathbb{P}_0)$. Then there exists a unique horizontal lift in $\mathcal{F}$ such that $L(0) = L_0$.
\end{proposition}
\begin{proof}
    Since the right-hand side of \cref{eq:lift_deq} is skew-Hermitian, $h(t) \in U(K) \; \forall t$. Suppose $\tilde{L}(t) = V(t)L(0)$ is a particular curve in $\mathcal{F}$ with the associated connection $\tilde{A} = \tilde{L}^\dagger \mathrm{d}\tilde{L}$. For the initial condition $h(0) = I_K$, the solution to \cref{eq:lift_deq} is 
    \[
        h(t) = \mathcal{P}\exp\left(-\int\tilde{A}\right),
    \]
    and hence there exists a unique horizontal lift.
\end{proof}

Finally, in the context of quantum computation on a stabilizer code, we wish to apply a logical unitary $\bar{U} \in U(K)$ on the code $\mathcal{C} = \mathbb{P}(0)$. Since the connection in \cref{eq:connection} is unique (up to the fiber) and the horizontal lift of the connection is also unique, applying a logical unitary is equivalent to finding a loop $\mathbb{P}(t) \in \mathcal{B}$, such that $\mathbb{P}(0) = \mathbb{P}(T)$, whose horizontal lift $L(t)$ produces the holonomy $\Gamma = \bar{U}$. In the following section, we show a method to apply this holonomy to a logical state using continuous measurements.

\section{Zeno-effect-induced holonomic quantum computation}\label{sec:mhqc}

Suppose we have a family of $[\![n, k, d]\!]$ quantum stabilizer codes that are parameterized by $\phi \in \{0, 2\pi\}$. Each code has the stabilizer group $S(\phi)$ with generators $\{g_i(\phi)\}_{i=1}^{n-k}$. The parameter $\phi=0$ denotes the standard code we operate on, for which the stabilizer generators are all Pauli operators, as usual. Without loss of generality, the system's Hilbert space can be decomposed into a direct sum of the simultaneous eigenspaces of all the stabilizers:
\[
    \mathcal{H}(\phi) = \bigoplus_{s=0}^{2^{n-k}} \mathcal{H}_s(\phi).
\]
By convention, the simultaneous $+1$ eigenspace is chosen to be the code space, denoted by $\mathcal{H}_0(\phi)$. In our case, all of the eigenspaces are $2^k$-fold degenerate. Define the projector onto the standard code space as
\[
    \mathbb{P}_0 \equiv \mathbb{P}(0)=\prod_{j=1}^{n-k}\frac{I+g_j(0)}{2}.
\]
The logical gate that we wish to apply is 
\[
    G=\mathrm{exp}\left(i\theta H\right),
\]
where $H$ is a \emph{logical} Pauli operator. We can generate a holonomy by rotating the code space by $\phi$. When the rotation forms a loop and the system is back in the original code space, the holonomy corresponds to a logical operation on the standard code.

To rotate the code space, we choose a Pauli operator $X$ (not to be confused by the single-qubit Pauli matrix $\sigma^x$) that satisfies the following conditions:
\begin{subequations}\label{eq:conditions_on_X_without_noise}
    \begin{align}
        \exists g \in S(0), \; \{X, g\}=0 , \\
        \{X, H\} = 0 , \\
        \exp\left(2\pi i X\right)=I.
    \end{align}
\end{subequations}
The rotation of the code space can be described using a family of unitary operators parameterized by $\phi$:
\[\label{eq:rotation_unitary}
    V(\phi) = \mathrm{exp}\left(i\frac{\theta\phi}{2\pi}H\right)\mathrm{exp}\left(i\phi X\right),
\]
and the projector onto the rotated code space is
\[\label{eq:proj_rotated_code}
    \mathbb{P}(\phi) = V(\phi)\mathbb{P}(0)V^\dagger (\phi).
\]
Observe that $V(0)=I$ and $V(2\pi)=\mathrm{exp}\left(i\theta H\right)$. So, starting with a state $\ket{\bar{\psi}}$ in the code space, in a purely theoretical sense, multiplying it by $V(2\pi)$ applies the logical transformation we desire. However, performing that in a real system is nontrivial; the goal of this work is to provide a method using measurements to realize this.

Consider a projective measurement of all the stabilizer generators of the unrotated code on a code state $\ket{\bar{\psi}}$. Since it is already an eigenstate, the probability of obtaining the same state $\ket{\bar{\psi}}$ is 1. Now suppose that the code space is rotated by the unitary operator $V(\phi)$. The stabilizer generators for this rotated code space are also rotated by the same unitary transformation:
\[
    g(\phi)=V(\phi)g(0)V^\dagger(\phi),
\]
where $g(0)$ is a stabilizer generator of the original (unrotated) code.

Consider the extreme case where a projective measurement of all the stabilizer generators of the ``fully'' rotated code (by $V(2\pi)$) is made on the original state $\ket{\bar{\psi}}$. This also returns the same state with unit probability, since $\mathbb{P}(2\pi) = \mathbb{P}_0$. In other words, since the logical operator $H$ commutes with all the stabilizers, and the rotation unitary operator at $\phi=2\pi$ depends only on the logical operator $H$ (and not at all on $X$), measuring the stabilizers of the fully rotated code is equivalent to measuring the stabilizers of the original code.

To effect the logical unitary we want, we rotate the code space by a \emph{small} amount, measure the generators of the rotated code, and repeat until $\phi=2\pi$ and we have returned to the original code space. To do this, let us decompose the unitary $V(2\pi)$ into a sequence of unitaries:
\[
    V(2\pi) = \tilde{U}(2\pi-\delta \phi)...\tilde{U}(\delta \phi)\tilde{U}(0),
\]
where $\tilde{U}(\phi)$ is defined as
\[
    \tilde{U}(\phi) \equiv V(\phi + \delta\phi)V^\dagger(\phi) .
\]
Intuitively, the unitary operator $\tilde{U}(\phi)$ rotates the code space $\mathbb{P}(\phi)$ by a small angle $\delta\phi$. Note that $\lim_{\delta\phi\rightarrow 0} \tilde{U}(\phi) = I$, i.e., the unitary is ``weak.'' Since $X$ anticommutes with at least one generator, the rotated generators do not have the same code space as the original space. Projectively measuring the rotated stabilizers instantaneously moves the state into an eigenstate. By the argument in \cref{sec:qze}, with high probability the post-measurement state is a simultaneous $+1$ eigenstate of the rotated stabilizers. After repeating the procedure until $\phi=2\pi$, the code state will have undergone an evolution through a \emph{holonomic path} and applied the logical gate $G$.
\begin{definition}[Holonomic path]
    The tuple $(V, \Phi)$ with entries denoting a family of unitary operators $V(\phi)$ and the corresponding final parameter $\Phi$, is a holonomic path along which the code space is rotated.
\end{definition}

\subsection{Discrete-path holonomy}\label{sec:proj_meas}

To perform the holonomic evolution, one can choose between measuring either the instantaneous code space projector or the instantaneous stabilizer generators. These measurements are not the same, but they both produce the desired holonomy. Moreover, in principle, it suffices to measure only one stabilizer generator.
\begin{proposition}
    Adiabatically measuring one rotated stabilizer that anticommutes with $X$ evolves the code space along the holonomic path $(V, \phi)$.
\end{proposition}
\begin{proof}
Suppose $\langle S(0)\rangle_1 = \{g_1, g_2, ... g_{n-k}\}$ is a generator set for the stabilizer group of the unrotated code such that a subset of the generators $S_{ac}(0) = \{g_1...g_l\}$ anticommute with $X$ ($S_c(0) = \{g_{l+1}...g_{n-k}\}$ commute with $X$). It is easy to see that every odd (even) product of generators from $S_{ac}(0)$ anticommutes (commutes) with $X$. This lets us choose a different generator set such that only one of the generators anticommutes with $X$, i.e., $\langle S(0)\rangle_2 = \{g_1, g_1g_2, g_1g_3, g_1g_l, g_{l+1} ... g_{n-k}\}$. The projector onto the simultaneous $+1$ eigenspace of the rotated stabilizers $\langle S(\phi) \rangle_1$ is
\[
\begin{split}
    \mathbb{P}(\phi) = &\left(\frac{I+V(\phi)g_1 V^\dagger(\phi)}{2}\right)...  \left(\frac{I+V(\phi)g_lV^\dagger(\phi)}{2}\right). \\
                       &\left(\frac{I+g_{l+1}}{2}\right)...  \left(\frac{I+g_{n-k}}{2}\right),
\end{split}
\]
where we used the fact that $[V(\phi), g_i] = 0$ for $i=l+1...n-k$. Rearranging the terms above, we obtain
\[
\begin{split}
    \mathbb{P}(\phi) = &\left(\frac{I+V(\phi)g_1 V^\dagger(\phi)}{2}\right).\left(\frac{I+g_1g_2}{2}\right)...\left(\frac{I+g_1g_l}{2}\right) \\
                       &\left(\frac{I+g_{l+1}}{2}\right)...  \left(\frac{I+g_{n-k}}{2}\right),
\end{split}
\]
which is another expression for that same projector onto the simultaneous $+1$ eigenspace of $\langle S(\phi) \rangle_2$.
\end{proof}

The ability to obtain the desired holonomy by just measuring a single rotating stabilizer offers interesting flexibility for practical implementation. However, we can also choose to measure all the stabilizers and use the syndrome information to do real-time quantum error correction. In this paper, we show a way to correct measurement-induced errors. Similar techniques can be used to correct any correctable error.

Suppose that we measure the code space projector. It is rotated using the unitary $V(\phi)$ in increments of $\delta\phi$. After every rotation, the new projector is measured; the state (up to a normalization factor) is multiplied by the rotated projector. Starting from the code state, we want to rotate the code space projector slowly enough so that the state stays in the $+1$ eigenspace with very high probability; this state can then be parameterized in terms of $\phi$ as the \emph{instantaneous code state} $\ket{\bar{\psi}(\phi)}$. If, on the other hand, the increment $\delta\phi$ is not small enough, a projective measurement of $\mathbb{P}(\phi+\delta\phi)$ can take the state $\ket{\bar{\psi}(\phi)}$ to the null space of $\mathbb{P}(\phi+\delta\phi)$. In other words, there is a nonzero probability for the code state to ``jump'' into the error space.

Define $p_\mathrm{jump}$ as the probability of a jump at any time during the procedure. For a given value of $p_\mathrm{jump}$, we can find the rotation increment $\delta\phi$ and hence the total number of measurements needed to complete the procedure and apply the logical gate. Observe that $1-p_\mathrm{jump}$ is the probability that the measurements always return $+1$ outcomes, in which case the unnormalized final state (denoted by $\ket{\tilde{\psi}(2\pi)}$) is obtained by multiplying the original code state by the sequence of projectors:
\[\label{eq:unnormalized_codestate}
\begin{split}
\ket{\tilde{\psi}(2\pi)} &= \mathbb{P}(2\pi)\mathbb{P}(2\pi-\delta\phi)\cdots\mathbb{P}(\delta\phi)\mathbb{P}_0\ket{\bar{\psi}}\\
&= [V(2\pi)\mathbb{P}_0V^\dagger(2\pi)]\cdots [V(\delta\phi)\mathbb{P}_0 V^\dagger(\delta\phi)]\mathbb{P}_0\ket{\bar{\psi}}\\
&= G \prod_{l=1}^{\ceil{2\pi/\delta\phi}} \mathbb{P}_0 V^\dagger\big(l\delta\phi\big)V\big((l-1)\delta\phi\big)\mathbb{P}_0 \ket{\bar{\psi}} .
\end{split}
\]
The normalization constant is proportional to $1-p_\mathrm{jump}$. \cref{eq:unnormalized_codestate} contains a product of operators of the form $\mathbb{P}_0 V^\dagger(\phi)V(\phi-\delta\phi)\mathbb{P}_0$. Using \cref{eq:rotation_unitary}, we find that this is proportional to the projection onto the original code space, followed by a rotation about $H$. In particular, we have the following lemma (which we will use later):
\begin{lemma}\label{lemma:small_rot_sandwiched_P0}
For a small rotation angle $\delta\phi$,
\[\label{eq:small_rot_sandwich_P0}
    \mathbb{P}_0V^\dagger(\varphi)V(\varphi-\delta\phi)\mathbb{P}_0 = c_\varphi e^{-i\xi_\varphi H}\mathbb{P}_0,
\]
where
\begin{subequations}
    \begin{align}
        \xi_\varphi &= \frac{\theta}{2\pi}\cos(2\varphi)\delta\phi + \mathcal{O}(\delta\phi^2) \\
        c_\varphi &= 1 - \frac{\delta\phi^2}{2}\left(1 + \frac{\theta^2}{4\pi^2}\sin^2(2\varphi)\right) + \mathcal{O}(\delta\phi^3).
    \end{align}
\end{subequations}
\end{lemma}
\begin{proof}
    See \cref{app:small_rot_sandwiched_P0_proof}.
\end{proof}

To find the normalized state, we multiply (from the left) the terms $\mathbb{P}_0V^\dagger(\varphi)V(\varphi-\delta\phi)\mathbb{P}_0$, where $\varphi=0\cdots 2\pi$. Note that the normalization constant $c_\varphi$ becomes irrelevant since the state after the measurement is normalized---the final state depends only on the phases $\xi_\varphi$:
\[
    \ket{\bar{\psi}(2\pi)} = G \exp\left(\sum_{l=1}^{\ceil{2\pi/\delta\phi}{}}\xi_{l\delta\phi}H\right)\ket{\bar{\psi}}.
\]
With elementary arithmetic, one can verify that
\[
    \sum_{l=1}^{\ceil{2\pi/\delta\phi}{}}\xi_{l\delta\phi} = 0,
\]
so that the final code state is indeed $G\ket{\bar{\psi}}$. We have the following theorem for the success probability:
\begin{theorem}
    Starting with a system in a code state, by measuring the code space projectors, rotated by small increments of $\delta\phi$, the probability that the state stays in the rotated code space throughout the loop is
    \[
        1 - p_\mathrm{jump} = \exp\left[-\frac{\delta\phi}{2\pi}\left(\frac{\theta^2}{2} + 4\pi^2\right)\right] + \mathcal{O}(\delta \phi^2).
    \]
\end{theorem}
\begin{proof}
    The probability that the state does not undergo a jump during the holonomic procedure is equal to the square of the normalization factor in $\ket{\tilde{\psi}(2\pi)}$. Using \cref{lemma:small_rot_sandwiched_P0}, we obtain
    \[
    \begin{split}
        1 - p_\mathrm{jump} &= \left(\prod_{l=1}^{\ceil{2\pi/\delta\phi}}c_{l\delta\phi}\right)^2 \\
        &= 1 - \frac{\delta\phi}{2\pi}\left(\frac{\theta^2}{2} + 4\pi^2\right) + \mathcal{O}(\delta\phi^2)\\
        &\approx \exp\left[-\frac{\delta\phi}{2\pi}\left(\frac{\theta^2}{2} + 4\pi^2\right)\right].
    \end{split}
    \]
\end{proof}

\begin{example}
    Suppose that we have the $[\![3, 1, 3]\!]$ bit-flip code with the code space $\mathcal{C} = \mathrm{span}\left\{\left(\ket{000}, \ket{111}\right)\right\}$, and the stabilizer generators $g_1 = \sigma_1^z\sigma_2^z, \; g_2 = \sigma_2^z\sigma_3^z$. Suppose further that we want to apply a logical $Z$ rotation. One unitary that can perform this is $G = \exp(i\theta \sigma_1^z\sigma_2^z\sigma_3^z)$, so $H=\sigma_1^z\sigma_2^z\sigma_3^z$. By the conditions mentioned in \cref{eq:conditions_on_X_without_noise}, we can choose $X = \sigma_3^x$, which anticommutes with both $H$ and the stabilizer generator $g_2$. With this choice, the stabilizer generator $g_1$ is unchanged, since it commutes with both $X$ and $H$. But the other generator $g_2(\phi)$ rotates nontrivially, as
    \[
    \begin{split}
        g_2(\phi) &= \cos(2\phi)g_2 - \sin(2\phi)\sin\left(\frac{\theta}{\pi}\phi\right)HXg_2 +\\
        &\qquad i\sin(2\phi)\cos\left(\frac{\theta}{\pi}\phi\right)Xg_2.
    \end{split}
    \]
    This operator still has eigenvalues $\pm 1$ and squares to the identity. Starting from a $+1$-eigenstate of $g_2$: $\ket{\bar{\psi}}$ and successively measuring $g_2(\phi)$ in small increments of $\delta\phi$ should, with high probability, always return the $+1$ outcome. The state after the final measurement of $g_2(2\pi)$ is $G\ket{\bar{\psi}}$.
\end{example}

If at any point during the procedure with $0 \leq \zeta \leq 2\pi$ we measured the projector $\mathbb{P}(\zeta)$ and got the result $0$; the state would then have jumped to $\mathrm{ker}\{\mathbb{P}(\zeta)\}$. The \emph{instantaneous error state} depends on $\zeta$, and it also includes a partial rotation about $H$, i.e., a logical error. In the following section, we show how to find this state and correct the logical error.

\subsection{Implicit fault-tolerance}\label{sec:fault_tolerance}

\begin{figure}
    \centering
    \includegraphics[width=\linewidth]{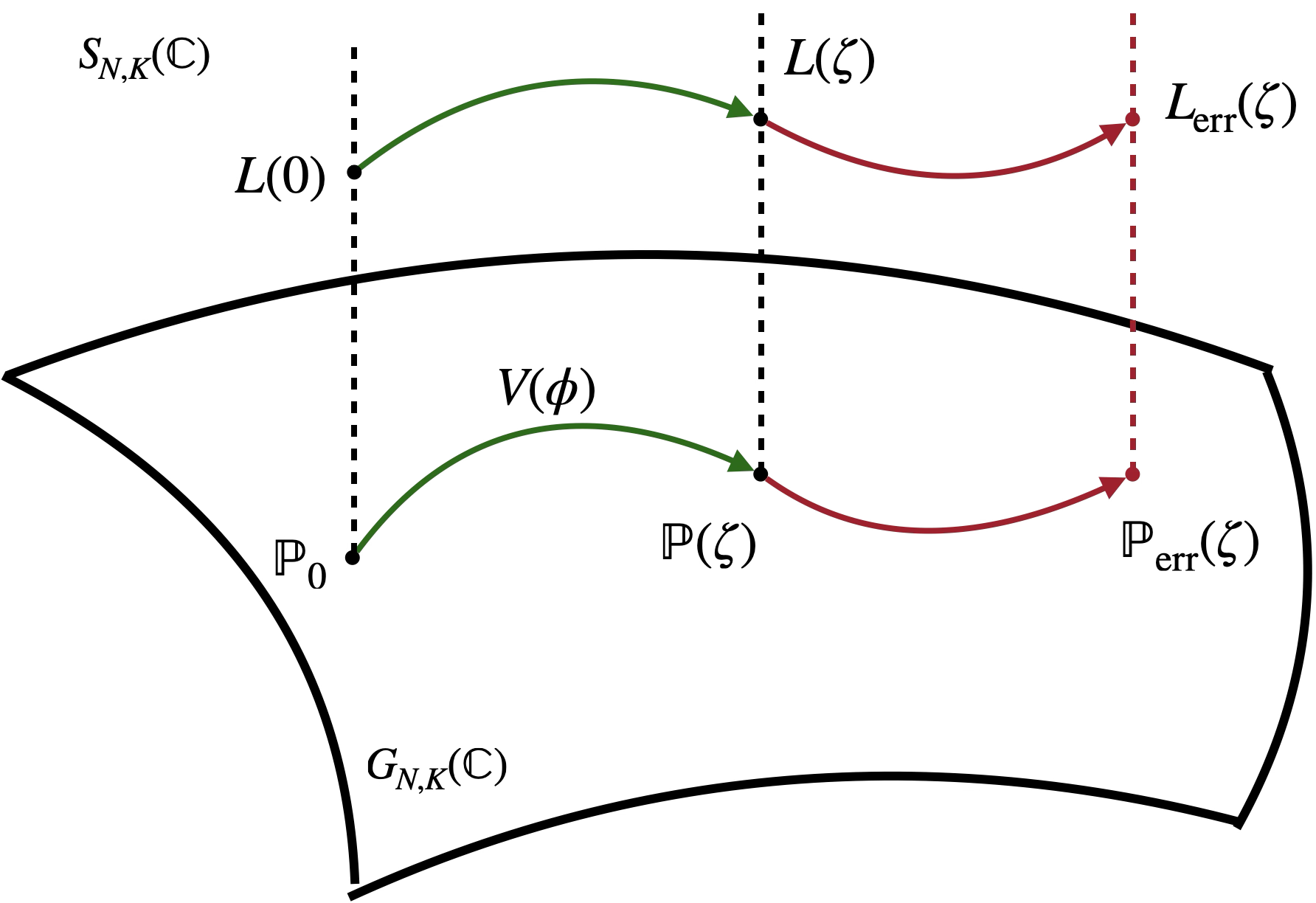}
    \caption{Jump into the rotated error space. When the code space is not rotated adiabatically, the state can jump into an error space defined by the operator $X$. The error state can potentially include a logical error.}
    \label{fig:mhqc_jump}
\end{figure}
During the holonomic evolution, if the stabilizer generators are not rotated slowly enough, the state can erroneously jump to the \emph{rotated} error space as illustrated in \cref{fig:mhqc_jump}. We find that the \emph{error state} following the jump contains a logical error. Using this information, we can correct the logical error and still achieve the desired holonomy by dynamically modifying the path. In this section, we do not consider the effects of external noise, i.e., the errors arise solely from measurement-induced jumps. For simplicity, we present a method for constructing a correction path that can correct an error due to a single jump. This construction can, in principle, be generalized to accommodate an arbitrary number of jumps. To begin, we first determine the rotated code state immediately before the jump:
\begin{proposition}\label{prop:instant_code_state}
    The instantaneous code state with the holonomic path $(V, \phi)$ is 
    \[\label{eq:inst_code_state}
        \ket{\bar{\psi}(\phi)} = V(\phi) \exp\left(-i \frac{\theta}{4\pi}\sin(2\phi) H\right)\ket{\bar{\psi}}.
    \]
\end{proposition}
\begin{proof}
    If the measurements of the instantaneous code space projectors always returned the outcomes $+1$, then the state always stays in their $+1$-eigenspaces. The curve on $\mathcal{B}$ can be written as
    \[
        \mathbb{P}(\phi) = V(\phi)\mathbb{P}_0V^\dagger(\phi),
    \]
    and the associated curve on $\mathcal{F}$ can be found by solving the differential equation \cref{eq:horizontal_lift}:
    \[\label{eq:lift_deq_path_V}
        \frac{\mathrm{d}h(\phi)}{\mathrm{d}\phi} = -L^\dagger(0)\left(V^\dagger(\phi)\frac{\mathrm{d}V(\phi)}{\mathrm{d}\phi}\right)L(0)h(\phi).
    \]
    Substituting \cref{eq:rotation_unitary} into \cref{eq:lift_deq_path_V}, we obtain
    \[
    \begin{split}
    \label{eq:deq_lift}
        \frac{\mathrm{d}h(\phi)}{\mathrm{d}\phi} &= -iL^\dagger(0)\left(\frac{\theta}{2\pi}H e^{i2\phi X} + X\right)L(0)h(\phi) \\
        & = -i\frac{\theta}{2\pi}\cos(2\phi)L^\dagger(0)HL(0)h(\phi) \\
        & \quad +\frac{\theta}{2\pi}\sin(2\phi)L^\dagger(0)HXL(0)h(\phi) \\
        & \quad -iL^\dagger(0)XL(0)h(\phi).
    \end{split}
    \]
    The columns of $L(0)$ represent the basis states of the original code space, i.e., they are code words. Since $X$ anticommutes with a stabilizer, we can treat $X$ as an ``error'' operator. Therefore, $XL(0)$ and $HXL(0)$ are the basis states of the orthogonal subspace, making $L^\dagger(0)XL(0) = L^\dagger(0)HXL(0) = 0$. This simplifies \cref{eq:deq_lift} to
    \[
    \begin{split}
    \label{eq:deq_lift_simplified}
        \frac{\mathrm{d}h(\phi)}{\mathrm{d}\phi} &= -i\frac{\theta}{2\pi}\cos(2\phi)L^\dagger(0)HL(0)h(\phi) \\
        &= -i F(\phi)h(\phi),
    \end{split}        
    \]
    where $F(\phi) \equiv (\theta/2\pi)\cos(2\phi)L^\dagger(0)HL(0)$. Since $[F(\phi), F(\phi')]=0 \; \forall \phi, \phi'$, the solution to \cref{eq:deq_lift_simplified} simplifies from a time-ordered exponential to a regular exponential:
    \[
    \begin{split}
        h(\phi) &= \exp\left(-i\frac{\theta}{2\pi}L^\dagger(0)HL(0)\int_{\phi'=0}^{\phi}\cos(2\phi')\mathrm{d}\phi'\right)h(0)\\
        &= \exp\left(-i\frac{\theta}{4\pi}\sin(2\phi)L^\dagger(0)HL(0)\right)h(0).
    \end{split}
    \]
    Without loss of generality, we can assume that the initial horizontal lift is trivial: 
    \[
    h(0) = I_{2^k}.
    \]
    Since $L^\dagger(0)L(0)=I_{2^k}$, we have 
    \[
        L^\dagger(0)e^AL(0) = \exp[L^\dagger(0)AL(0)]
    \]for any operator $A$. Then
    \[
        h(\phi) = L^\dagger(0)\exp\left(-i\frac{\theta}{4\pi}\sin(2\phi)H\right)L(0),
    \]
    and the curve $L(\phi)$ can be written as 
    \[
        L(\phi) = V(\phi)\exp\left(-i\frac{\theta}{4\pi}\sin(2\phi)H\right)L(0),
    \]
    where we used the fact that $L(0)L^\dagger(0) = \mathbb{P}_0$ commutes with $H$. Using \cref{eq:final_state_with_vertical_lift}, the state at $\phi$ can be written
    \[
        \ket{\bar{\psi}(\phi)} = V(\phi)\exp\left(-i\frac{\theta}{4\pi}\sin(2\phi)H\right)\ket{\bar{\psi}}.
    \]
    An alternative proof of this result is given in \cref{sec:inst_estate}.
\end{proof}
This is the $+1$-eigenstate of the projector onto the rotated code space. When there is no jump to the error space throughout the protocol, we obtain $\ket{\bar\psi(2\pi)} = G\ket{\bar{\psi}}$, and the code state undergoes the desired logical gate. However, with probability $p_\mathrm{jump}$, the state jumps into the error space. In the absence of external noise, the system state can only evolve into the ``error'' space  defined by the operator $X$, with the projector
\[\label{eq:error_space}
    \mathbb{P}_\mathrm{err}(\zeta) = V(\zeta)X\mathbb{P}_0XV^\dagger(\zeta).
\]
Moreover, the error operator that moves an instantaneous code state to the instantaneous error state is unitary:
\begin{proposition}
    In the absence of external noise, the sole source of error arises from the non-adiabatic rotations of the stabilizer generators; the resulting phase-dependent error is
    \[
        E(\zeta) = V(\zeta)Xe^{-i\chi(\zeta)H}V^\dagger(\zeta),
    \]
    where
    \[
     \chi(\zeta) = \mathrm{arctan}\left(\frac{\theta}{2\pi}\sin(2\zeta)\right).
    \]
\end{proposition}
\begin{proof}
Suppose that the rotation at $\zeta$ with $0 \leq \zeta \leq 2\pi$ resulted in an error, thereby projecting the state into the error space $V(\zeta)(I-\mathbb{P}_0)V^\dagger(\zeta)$. The instantaneous error state can be found by multiplying the original code state through the sequence of rotated projectors until $\zeta-\delta\phi$, followed by the orthogonal complement at $\zeta$:
\[
\begin{split}
    \ket{\psi_\mathrm{err}(\zeta)} &= (I-\mathbb{P}(\zeta))\cdot \mathbb{P}(\zeta-\delta\phi)\cdots\mathbb{P}_0\ket{\bar{\psi}}\\
        &= [V(\zeta)(I-\mathbb{P}_0) V^\dagger(\zeta)]\cdots[V(\delta\phi)\mathbb{P}_0V^\dagger(\delta\phi)]\ket{\bar{\psi}}.
\end{split}
\]
Until $\zeta-\delta\phi$, the state undergoes the error-free evolution described in \cref{eq:inst_code_state}. At $\zeta$, the instantaneous code state $\ket{\bar{\psi}(\zeta-\delta\phi)}$ is multiplied by
\begin{subequations}
\begin{align}
    &(I - \mathbb{P}_0) V^\dagger(\varphi) V(\varphi - \delta\phi) \mathbb{P}_0 = \notag \\
    &\quad\qquad \sin\left(\frac{\theta}{2\pi} \delta\phi\right) \sin(2\zeta - \delta\phi) H X \mathbb{P}_0 \\
    &\qquad\qquad\qquad- i \cos\left(\frac{\theta}{2\pi} \delta\phi\right) \sin(\delta\phi) X \mathbb{P}_0 \\
    &\quad\equiv d(\zeta) e^{i\xi(\zeta) H} X \mathbb{P}_0.
\end{align}
\end{subequations}
The normalization constant is given by
\[
\begin{split}
    d(\zeta) &= 
    \left[
        \sin^2\left(\frac{\theta}{2\pi} \delta\phi\right) \sin^2(2\zeta-\delta\phi) + \right. \\
        & \left. \qquad\qquad\cos^2\left(\frac{\theta}{2\pi} \delta\phi\right) \sin^2(\delta\phi)
    \right]^{1/2},
    \end{split}
\]
and the phase $\xi(\zeta)$ follows
\[
\begin{split}
    \tan(\xi(\zeta)) &= \frac{\sin\left(\frac{\theta}{2\pi} \delta\phi\right) \sin(2\zeta - \delta\phi)}{\cos\left(\frac{\theta}{2\pi} \delta\phi\right) \sin(\delta\phi)} \\
    &= \frac{\theta}{2\pi}\sin(2\zeta-\delta\phi)+\mathcal{O}(\delta\phi^2).
\end{split}
\]
The evolved state in the error space at $\zeta$ is
    \[
    \begin{split}
        \ket{\psi_\mathrm{err}(\zeta)} &= V(\zeta)X\exp\left(-i\arctan\left[\frac{\theta}{2\pi}\sin(2\zeta-\delta\phi)\right]H\right) \\
        &\qquad\qquad\exp\left(i\frac{\theta}{2\pi}\delta\phi H\right)\exp\left(-i\frac{\theta}{4\pi}\sin(2\zeta)H\right)\ket{\bar{\psi}}\\
        &= \left[V(\zeta)\exp(i\chi(\zeta) H)X V^\dagger(\zeta)\right] \ket{\psi(\zeta)} + \mathcal{O}(\delta\phi),
    \end{split}
\]
with $\chi(\zeta) \equiv \arctan\left(\frac{\theta}{2\pi}\sin(2\zeta)\right)$ in the limit $\delta\phi \rightarrow 0$. Then we have the following error operator at $\zeta$:
\[\label{eq:error_op}
    E(\zeta) = V(\zeta)\exp(i\chi(\zeta) H)X V^\dagger(\zeta).
\]
While we used the projector $I-\mathbb{P}(\zeta)$ to denote the orthogonal subspace when there is a jump at $\zeta$, the projector can also be obtained from \cref{eq:error_space}. One can verify that the error state obtained will be the same.
\end{proof}

Assume that the error has occurred at $\phi = \zeta$. Then the state in the error space is
\[\label{eq:state_after_jump}
\begin{split}
    \ket{\psi_\mathrm{err}(\zeta)} &= E(\zeta)\ket{\bar{\psi}(\zeta)} \\
    &= V(\zeta)\exp\left(i\theta_\mathrm{err}(\zeta)H\right))X\ket{\bar{\psi}},
\end{split}
\]
where the erroneous rotation angle $\theta_\mathrm{err}(\phi)$ is defined as
\[
    \theta_\mathrm{err}(\phi) \equiv \frac{\theta}{4\pi}\sin(2\phi) + \chi(\phi).
\]
The measurement-induced jump introduces a ``fault'' through a two-step process: an application of $X$, followed by a rotation of the code state about $H$ by $\theta_\mathrm{err}(\zeta)$, introducing a logical error. This fault is a direct consequence of non-adiabatic rotation of the code space; achieving fault-tolerance then requires suppressing the probability of these faults. While slow rotation can reduce this probability, it lengthens the protocol and may potentially diminish any quantum advantage. Alternatively, since we know the state after the jump, we can perform error correction. This can be done either by reversing both the $X$ operation and the logical error, or by correcting only the error while tracking the Pauli frame. In the latter case, subsequent computation proceeds within the error space. Since we are continuously measuring the stabilizers throughout this protocol, we can detect a jump by tracking the expectation of the stabilizer that anticommutes with $X$---its value will flip to $-1$ when the jump occurs. We can then correct this error simply by restarting the holonomic procedure after the jump, but with a different path.
\begin{figure}
    \centering
    \includegraphics[width=\linewidth]{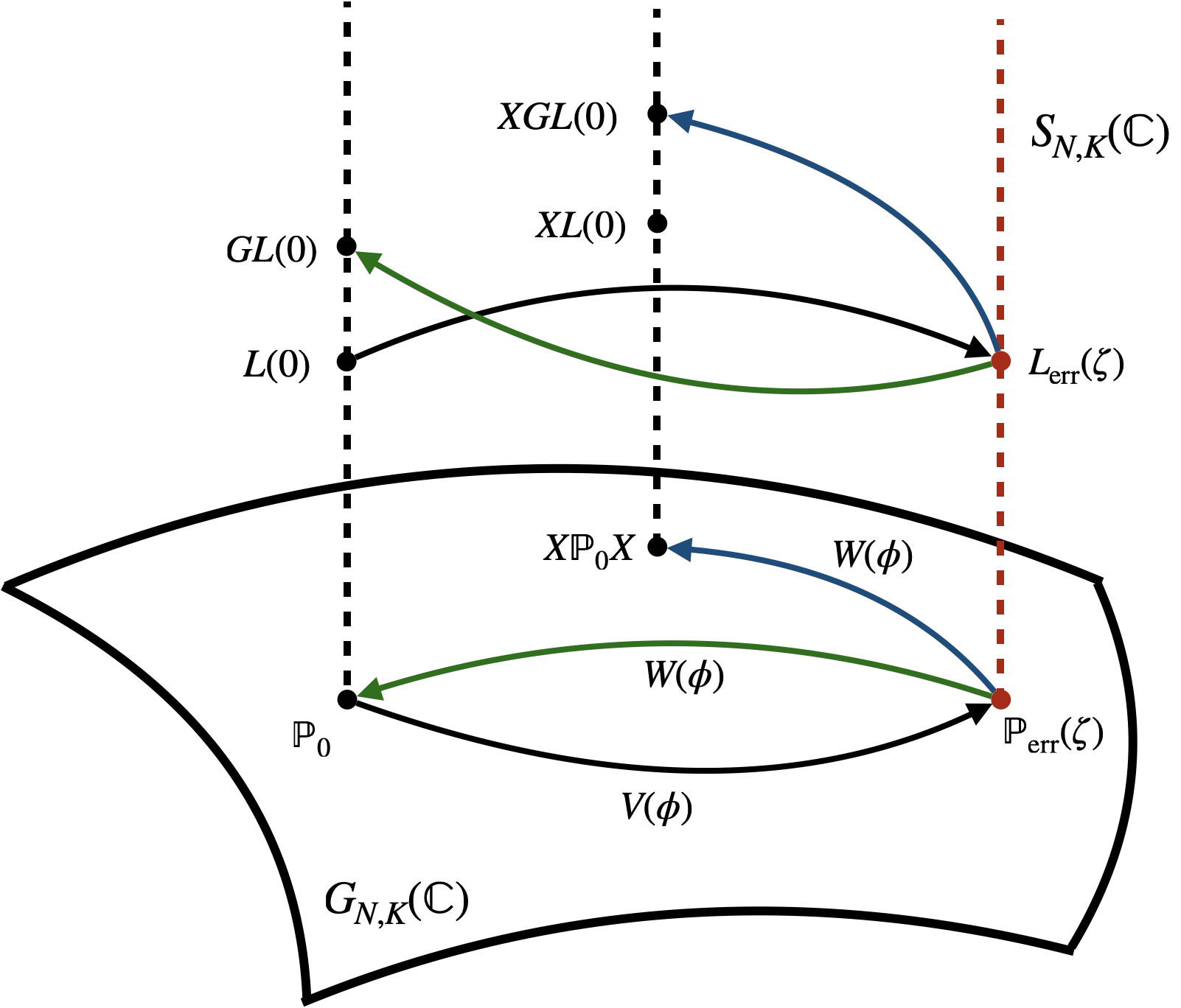}
    \caption{Real-time error correction by dynamic path modulation. When non-adiabatic evolution results in a measurement-induced error, it can be corrected by changing the holonomic path to either $(W, 7\pi/2-\zeta)$ or $(W, 4\pi-\zeta)$. In the former case, the state returns to the original code space with the desired holonomy, denoted by the green curves. In the latter case, the state returns to the original error space with the emulated holonomy, denoted by the blue curves. The operator $X$ can be corrected or stored in a Pauli frame and the subsequent computation can be performed.}
    \label{fig:error}
\end{figure}

\begin{theorem}\label{thm:new_path}
Suppose that the holonomic procedure along the path $(V, \zeta)$ resulted in a measurement error. Define a family of unitary operators parameterized by $\phi$:
\[\label{eq:new_path_op}
    W(\phi) = \exp\left(-i\tilde{\theta}\frac{\phi}{2\pi}H\right)V(\phi+\zeta)V^\dagger(\zeta).
\]
By modifying the path to $(W, 7\pi/2 - \zeta)$ with
\[\label{eq:new_angle_codespace}
    \tilde{\theta} = \frac{3\pi\theta - 4\pi\theta_{\mathrm{err}}(\zeta) - \theta\sin(2\zeta)}{7\pi - 2\zeta - \sin(2\zeta)},
\]
and restarting the holonomic procedure, the resulting state is in the original code space with the desired holonomy. Alternatively, modifying the path to $(W, 4\pi-\zeta)$ with
\[\label{eq:new_angle_errorspace}
    \tilde{\theta} = \frac{12\pi\theta + 4\pi\theta_{\mathrm{err}}(\zeta) + \theta\sin(2\zeta)}{8\pi - 2\zeta + \sin(2\zeta)}
\]
results in a state is in the original error space with the desired emulated holonomy: $\ket{\psi(4\pi-\zeta)}= XG\ket{\bar{\psi}}$.
\end{theorem}
\begin{proof}
    When a jump occurs at $\zeta$, the state is projected to a subspace with projector $\mathbb{P}_\mathrm{err}(\zeta)$. Let the path be modified to $W(\phi)$. It is easy to see that $W(0)=I$, so at the initial rotation angle \emph{after} the jump, $W(0)\mathbb{P}_\mathrm{err}(\zeta)W^\dagger(0) = \mathbb{P}_\mathrm{err}(\zeta)$. Moreover, $W(\phi+\delta\phi)W^\dagger(\phi) = I+\mathcal{O}(\delta\phi)$, i.e., the rotation is weak. We would like to run the procedure for at least $2\pi$, so that the state $\ket{\psi_\mathrm{err}(\zeta)}$ is brought back to the desired code space $\mathbb{P}_0$ (or to the error space $X\mathbb{P}_0X$) with a jump probability of at most $p_\mathrm{jump}$. We restart the procedure, resetting the parameter to $\phi=0$, and run until $\phi=7\pi/2-\zeta$ (respectively, $\phi = 4\pi-\zeta$) to bring the state back to the code (error) space. At the final rotation angle,
    \[
        W(7\pi/2-\zeta) = -i\exp\left[i\left(\frac{7}{4}(\theta - \tilde{\theta}) + \frac{\tilde{\theta}\zeta}{2\pi}\right)H\right]XV^\dagger(\zeta),
    \]
    or
    \[
        W(4\pi-\zeta) = \exp\left(-i\tilde{\theta}\left(\frac{4\pi-\zeta}{2\pi}\right)H\right)\exp\left(i2\theta H\right)V^\dagger(\zeta).
    \]
    If the state is brought back to the error space, the unitary transformation is not strictly a holonomy; it emulates the effective holonomy in the error space.

    Let the curve after the jump be represented by $L'(\phi)$. From \cref{eq:state_after_jump}, we can write 
    \[
    L'(0) = L_\mathrm{err}(\zeta) \equiv V(\zeta)X\exp\left(-i\theta_\mathrm{err}(\zeta)H\right)L(0).
    \]
    Using the path $W(\phi)$, this curve can be written as
    \[
        L'(\phi) = W(\phi)L'(0)h'(\phi).
    \]
    The dynamics of the horizontal lift follow:
    \[
    \begin{split}
        \frac{\mathrm{d}h'(\phi)}{\mathrm{d}\phi} =& -L'^\dagger(0)W^\dagger(\phi)\frac{\mathrm{d}W(\phi)}{\mathrm{d}\phi}L'(0)h'(\phi) \\
        =& -L^\dagger(0)\exp(i\theta_\mathrm{err}(\zeta)H) \\ &\times X\left(\tilde{V}^\dagger(\phi)\frac{\mathrm{d}\tilde{V}(\phi)}{\mathrm{d}\phi}\right)X\\
        &\times \exp(-i\theta_\mathrm{err}(\zeta)H)L(0)h'(\phi),
    \end{split}
    \]
    where
    \[
        \tilde{V}(\phi) \equiv \exp\left(i\frac{1}{2\pi}(\theta\zeta + (\theta-\tilde{\theta})\phi)H\right) \exp\left(i(\zeta+\phi)X\right).
    \]
    Using the fact that $L^\dagger(0)XL(0) = L^\dagger(0)HXL(0) = 0$, this can be written as
    \[
        \frac{\mathrm{d}h'(\phi)}{\mathrm{d}\phi} = i\cos\left(2(\zeta+\phi)\right)\left(\frac{\theta-\tilde{\theta}}{2\pi}\right)L^\dagger(0)HL(0)h'(\phi),
    \]
    with a solution
    \[
        h'(\phi) = L^\dagger(0)\exp\left(-i\left(\frac{\theta-\tilde{\theta}}{4\pi}\right)\sin(2\zeta)H\right)L(0)h'(0).
    \]
    \begin{widetext}
    By assumption, we have $h'(0) = I_K$, and the final point on the curve (if moving to code space) is
    \[
    \begin{split}
        L'(7\pi/2-\zeta) &= W(7\pi/2-\zeta)L'(0)h'(7\pi/2-\zeta) \\
        &= -i\exp\left(i\left[\frac{7}{4}(\theta-\tilde{\theta}) +\frac{\tilde{\theta}\zeta}{2\pi} -\theta_\mathrm{err}(\zeta) - 
        \left(\frac{\theta-\tilde{\theta}}{4\pi}\right)\sin(2\zeta)\right]H\right)L(0) .
    \end{split}
    \]
    Comparing this with the target point on the curve $L'(7\pi/2-\zeta) = e^{i\theta H}L(0)$, we have
    \[
        \tilde{\theta} = \frac{3\pi\theta - 4\pi\theta_{\mathrm{err}}(\zeta) - \theta\sin(2\zeta)}{7\pi - 2\zeta - \sin(2\zeta)}.
    \]
    On the other hand, if the state is to be moved to the error space, the final point on the curve is
    \[
    \begin{split}
        L'(4\pi-\zeta) &= W(4\pi-\zeta)L'(0)h'(\phi) \\
        &= X \exp\left(i\left[\left(\frac{4\pi-\zeta}{2\pi}\right)\tilde{\theta}-2\theta-\theta_\mathrm{err}(\zeta) - \left(\frac{\theta-\tilde{\theta}}{4\pi}\right)\sin(2\zeta)\right]H\right)L(0) .
    \end{split}
    \]
    \end{widetext}
    Comparing this with the target point on the curve $L'(4\pi-\zeta) = Xe^{i\theta H}L(0)$, we have
    \[
       \tilde{\theta} = \frac{12\pi\theta + 4\pi\theta_{\mathrm{err}}(\zeta) + \theta\sin(2\zeta)}{8\pi - 2\zeta + \sin(2\zeta)}.
    \]
\end{proof}

When a jump happens at $\zeta$, we have $2$ degrees of freedom in path selection. Specifically, the holonomic path is modified to $(W, 7\pi/2-\zeta)$ (respectively, $(W, 4\pi-\zeta)$) when transitioning to the code (error) space. When returning to the error space, the unitary transformation---while not strictly holonomic---emulates the effective holonomy within the original error space $X\mathbb{P}_0 X$. Note that this technique to modify the path only accounts for a single jump. We can, in principle, repeat this to find the updated paths when there are multiple jumps, so that the state reaches the target without any faults.
\begin{example}
    Consider the $[\![3, 1, 3]\!]$ bit-flip code with $H=\otimes_{i=1}^3\sigma_i^x$, $\theta=\pi/6$ and $X = \sigma_1^x\sigma_3^z$. We compared the probability of no faults when the path is kept constant at $V(\phi)$ throughout the evolution, with the one where the path is changed once at the first jump as shown in \cref{fig:ft_sims}. The plots shown are averaged over an ensemble of $1000$ trajectories. The error bars represent a $95\%$ confidence interval.
\end{example}

In this paper, we do not consider the influence of external noise, which will also impact the path operator, but we show in another work that such errors can also be corrected by modifying the path \cite{lanka2026steeringpathsmidflightfaulttolerance}.

\begin{figure}
    \centering
    \includegraphics[width=\linewidth]{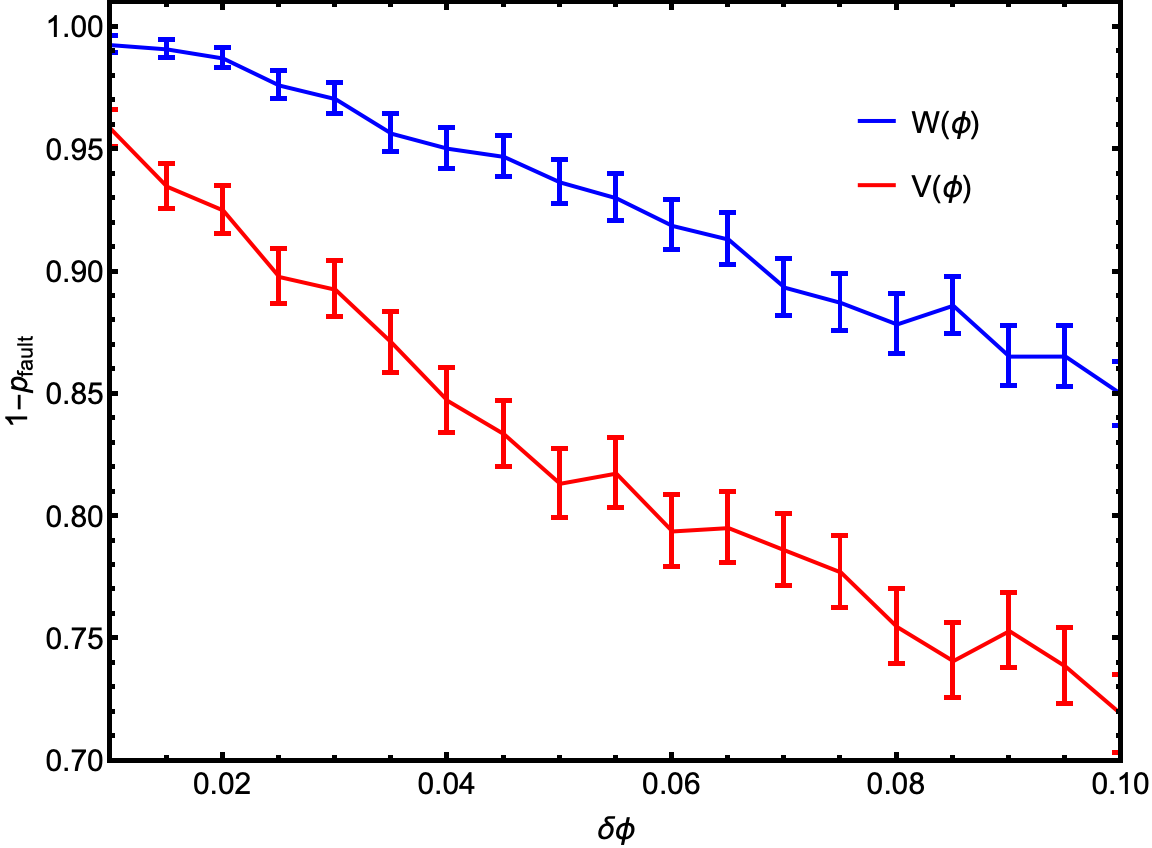}
    \caption{Probability of no fault at the end of the evolution as function of the rotation increment $\delta\phi$ for the $[\![3, 1, 3]\!]$ bit-flip code with $H=\otimes_{i=1}^3\sigma_i^x$, $\theta=\pi/6$ and $X = \sigma_1^x\sigma_3^z$. The plots are averaged over an ensemble of $2000$ trajectories. The error bars represent a $95\%$ confidence interval.}
    \label{fig:ft_sims}
\end{figure}

\subsection{Continuous-path holonomy}\label{sec:continuous_meas}

The holonomic procedure described in \cref{sec:proj_meas} involves making projective measurements after every small rotation by $\delta\phi$. Although that should work in principle, it is hard (and somewhat unnecessary) to perform a sequence of strong projective measurements in rapid succession. Since the rotation angle is small, we can perform continuous weak measurements of the stabilizers while also simultaneously rotating them. To do that, we define the ``angular velocity''
\[
    \omega = \frac{\mathrm{d}\phi}{\mathrm{d}t}
\]
as the \emph{constant} rate of code space rotation. Then, we can re-parameterize \cref{eq:rotation_unitary} as
\[
    V(t) = \mathrm{exp}\left(i\omega_H t H\right)\mathrm{exp}\left(i\omega_X t X\right),
\]
where $\omega_X \equiv \omega$, $\omega_H \equiv \frac{\theta}{2\pi} \omega$, and accordingly 
\[
    g(t)=V(t)g(0)V^\dagger(t).
\]
The dynamics of a state undergoing continuous weak measurements of the rotating stabilizers are described by a time-dependent stochastic Schr\"{o}dinger equation:
\[\label{eq:sse}
\begin{aligned}
    \ket{\mathrm{d}\psi} = -\frac{\kappa}{2}\sum_j\big(&g_j(t) - \langle g_j(t)\rangle \big)^2 \ket{\psi} \mathrm{d}t \\
                &+ \sqrt{\kappa} \big(g_j(t) - \langle g_j(t)\rangle \big)\ket{\psi}\mathrm{d}W_t^{(j)},
\end{aligned}
\]
where $\kappa$ is the measurement strength and $\mathrm{d}W_t^{(j)}$ is the Wiener increment for the $j^\text{th}$ measurement, defined by a zero-mean Gaussian random variable with variance $\mathrm{d}t$. Starting in a code state, we continuously measure the stabilizers, which are uniformly rotating from $\phi=0$ to $2\pi$ at a rate $\omega$. The procedure is run for a total time of $T=2\pi/\omega$, at which point the code space has rotated through a loop and back to the original code space and the code state should have acquired the desired holonomy. In principle, this rotation rate should depend on the measurement rate: the measurement has to be made faster than the code space is rotated, so that the state lies in the $+1$ eigenspace at all times with high probability. However, we have found that in practice the ratio $\omega/\kappa$ cannot be arbitrarily small since for a given measurement strength, very slow rotations correspond to very long gate times, which may diminish any potential quantum advantage. Therefore, the choice of $\omega$ is crucial: the rotations must be neither too fast, which could cause the state to jump into an error space, nor too slow, which would prolong the gate time. To find a suitable $\omega$, we determine the probability of avoiding a jump into an error space for a given $\omega/\kappa$. Observe that the channel corresponding to the rotated stabilizer measurements (and forgetting the outcomes) is unital---it is purity non-increasing. Hence, to find the average probability of measuring the state in the code space after a finite time $T$, it is helpful to represent the dynamics in terms of a density matrix, rather than a state vector, even though the system is governed purely by measurements. In particular, we have $\rho = \mathbb{E}[\ket{\psi}\bra{\psi}]$, yielding the Lindblad equation for the average dynamics:
\[
    \mathrm{d}\rho = \kappa \sum_j \big[g_j(t)\rho g_j(t) - \rho\big] \mathrm{d}t.
\]
The analysis is easier in the rotating frame, defined by $\tilde{\rho} = V^\dagger(t) \rho V(t)$. The dynamics for $\tilde{\rho}$ are given by
\[\label{eq:dynamics_rot_frame}
    \frac{\mathrm{d}\tilde{\rho}}{\mathrm{d}t} = -i[H_{\text{eff}}(t), \tilde{\rho}] + \kappa \sum_j(g_j\tilde{\rho}g_j - \tilde{\rho}),
\]
where 
\[
    H_{\text{eff}}(t) \equiv \omega_X X + \omega_H \tilde{H}(t)
\]
is the \emph{effective Hamiltonian}, and 
\[
    \tilde{H}(t) \equiv \cos(2\omega_X t)H - i\sin(2\omega_X t)XH.
\]

As with the projective measurements case, we would like to find the probability that the state makes no jump when the evolution is driven by performing weak measurements. Note that the state should have subject to at least one jump when the expectation of the original code space projector deviates from unity. The probability that the state makes no jump can defined as
\[
\begin{split}
    1 - p_\mathrm{jump} &\equiv \langle \mathbb{P}_0\rangle_{\rho(T)} \\
    &= \frac{1}{2^r} \sum_{N \subseteq \{1, 2, ... r\}} \prod_{j \in N} \langle g_j(T) \rangle_{\rho(T)},
\end{split}
\]
where the sum is over all subsets $N$ of indices, including the empty set. (The empty product (when $N = \emptyset$) is defined to be $1$.) Observe that the expectations of the rotated stabilizers with respect to the state in the lab frame are equal to those of the unrotated stabilizers with respect to the state in the rotated frame: 
\[
\langle g_j(t) \rangle_{\rho(t)} = \langle g_j \rangle_{\tilde{\rho}(t)}.
\]
Let us find the expectation of a stabilizer $g$ that anticommutes with $X$. (We omit the subscript $\tilde{\rho}$ for brevity, but all the expectations are taken with respect to $\tilde{\rho}$.) Consider the dynamics of $\langle g \rangle$:
\[
    \begin{aligned}
        \frac{\mathrm{d}\langle g \rangle}{\mathrm{d}t} &= \text{Tr}\left[\frac{\mathrm{d}\tilde{\rho}}{\mathrm{d}t}g\right] \\
        &= -2i\omega_X\langle gX \rangle - 2\omega_H\sin(2\omega_X t)\langle gXH \rangle.
    \end{aligned}
\]
To get a closed set of equations, we also find the dynamics of $\langle gX \rangle$ and $\langle gXH \rangle$:
\begin{subequations}
    \begin{align}
        &\begin{split}
            \frac{\mathrm{d}\langle gX \rangle}{\mathrm{d}t} = 
            &-2i\omega_X\langle g \rangle \\ 
            &- 2i\omega_H\cos(2\omega_X t)\langle gXH \rangle \\
            &- 2\kappa \langle gX \rangle,
        \end{split} \\
        &\begin{split}
            \frac{\mathrm{d}\langle gXH \rangle}{\mathrm{d}t} = & -2i\omega_H\cos(2\omega_H t) \langle gX \rangle \\
            &+ 2\omega_H\sin(2\omega_X t)\langle g \rangle \\
            &- 2\kappa \langle gXH \rangle.
        \end{split}
    \end{align}
\end{subequations}
To solve for $\langle g \rangle$, we solve a coupled differential equation for the vector $\mathbf{x}_g(t) = \left(\langle g\rangle,\langle gX\rangle,\langle gXH\rangle\right)^T$: 
\[\label{eq:coup_diff_eq}
    \frac{\mathrm{d}\mathbf{x}_g(t)}{\mathrm{d}t}=A(t)\mathbf{x}_g(t),
\]
where the time-dependent matrix $A(t)$ is given by
\[
\begin{split}
    &A(t) = \\
    &\begin{pmatrix}
    0 & -2i\omega_X & -2\omega_H \sin(2\omega_X t) \\
    -2i\omega_X & -2\kappa & -2i \omega_H \cos(2\omega_X t) \\
    2\omega_H \sin(2\omega_X t) & -2i \omega_H \cos(2\omega_X t) & -2\kappa
\end{pmatrix}.
\end{split}
\]
Consider the following lemma on the evolution of the expectation values of stabilizer generators based on their structure:
\begin{lemma}\label{lemma:prod_stabs_evol}
    The expectations of all stabilizer generators that anticommute with $X$ evolve symmetrically in the rotated frame: $\langle g \rangle_{\tilde{\rho}(t)} = \langle h \rangle_{\tilde{\rho}(t)} \; \forall g, h \in S_\mathrm{ac}(0)$. The expectations of all the stabilizer generators that commute with $X$ remain constant at 1: $\langle g \rangle_{\tilde{\rho}(t)} = 1 \; \forall g \in S_\mathrm{c}(0)$.
\end{lemma}
\begin{proof}
The expectation of $g$ that anticommutes with $X$ can be found by integrating \cref{eq:coup_diff_eq}:
\[
    \textbf{x}_g(t) = \mathcal{T}\exp\left[\int_0^tA(t)\mathrm{d}t\right]\textbf{x}_g(0),
\]
where $\mathcal{T}$ is the time-ordering operator. Similarly, for different stabilizer $h$, we have 
\[
    \textbf{x}_h(t) = \mathcal{T}\exp\left[\int_0^tA(t)\mathrm{d}t\right]\textbf{x}_h(0),
\]
Since the system is initialized in the code space, the expectation values of all the stabilizers are equal to $1$, and the expectation of $gX$, and $gXH$ for any stabilizer $g$ that anticommutes with $X$ are 0; in other words, $\textbf{x}_g(0) = \textbf{x}_h(0)$, giving $\langle g \rangle = \langle h \rangle$.

On the other hand, if the stabilizer $g$ commutes with $X$, it is easy to see that $\mathrm{d}\langle g \rangle = 0$, i.e., its expectation value remains constant at the initial value of $+1$.
\end{proof}

Observe that an even product of the stabilizer generators that anticommute with $X$ (say, $g_\text{even}$) commutes with $X$. By \cref{lemma:prod_stabs_evol}, $\langle g_\mathrm{even} \rangle (t) = 1\, \forall t$. However, the expectation of an odd product of the stabilizers $g_\mathrm{odd}$ evolves nontrivially as $\langle g_\mathrm{odd} \rangle_{\tilde{\rho}(t)}$. Again, by \cref{lemma:prod_stabs_evol}, all the nontrivial expectations can be combined into $\langle g \rangle_{\tilde{\rho}(t)}$, yielding
\[
    1 - p_\mathrm{jump} = \frac{1+\langle g \rangle_{\tilde{\rho}(t)}}{2}.
\]
Since we generally operate in a regime where the rotations are slower than measurements, we can use perturbation theory to solve \cref{eq:coup_diff_eq} up to second order in $\omega/\kappa$.
\begin{theorem}
Suppose a stabilizer code $\mathcal{C}$ is continuously rotated at a rate $\omega$ by $V(t)$. Starting at $\ket{\bar{\psi}(0)} \in \mathcal{C}$, by continuously measuring the rotating stabilizer generators at a rate $\kappa$, the probability of no jumps between the syndrome spaces is
\[\label{eq:confinement_prob}
    1 - p_{\mathrm{jump}} = \frac{1}{2}\bigg[ 1 + \bigg(1 - \frac{\omega \theta^2}{2\pi\kappa}\bigg)\exp\bigg(-\frac{4\pi\omega}{\kappa}\bigg) \bigg].
\]
\end{theorem}
\begin{proof}
Suppose $A(t)$ is decomposed into a time-independent matrix $A_0$ and a time-dependent matrix $A_1(t)$ with $A(t) = A_0 + \omega_H A_1(t)$:
\begin{subequations}\label{eq:A_t}
\begin{align}
    A_0 & = \begin{pmatrix}
        0 & -2i\omega_X & 0 \\
        -2i\omega_X & -2\kappa & 0 \\
        0 & 0 & -2\kappa
    \end{pmatrix} \\
    A_1(t) &= \begin{pmatrix}
        0 & 0 & -2\sin(2\omega_X t) \\
        0 & 0 & -2i \cos(2\omega_X t) \\
        2\sin(2\omega_X t) & -2i \cos(2\omega_X t) & 0
    \end{pmatrix},
\end{align}
\end{subequations}
and consider the following ansatz for the solution $\mathbf{x}(t)$:
\[\label{eq:ansatz}
    \mathbf{x}(t) = \sum_{j=0}^\infty \tilde{\mathbf{x}}^{(j)}(t),
\]
where $\tilde{\mathbf{x}}^{(j)}(t) = \omega_H^j \mathbf{x}^{(j)}(t)$. Substituting \cref{eq:ansatz} in \cref{eq:coup_diff_eq} and collecting the terms by the degree of their coefficients, we obtain
\begin{subequations}
\begin{align}
    &\frac{\mathrm{d}}{\mathrm{d}t} \tilde{\mathbf{x}}^{(0)}(t) = \mathbf{A}_{0}\mathbf{x}^{(0)}(t), \\
    &\frac{\mathrm{d}}{\mathrm{d}t} \tilde{\mathbf{x}}^{(n)}(t) = A_0 \tilde{\mathbf{x}}^{(n)}(t) + \omega_H A_1(t) \tilde{\mathbf{x}}^{(n-1)}(t) \;\;\; n \geq 1, 
\end{align}
\end{subequations}
with solutions
\begin{subequations}\label{eq:perturb_sol}
\begin{align}
    \tilde{\mathbf{x}}^{(0)}(t) &= e^{A_0 t}\tilde{\mathbf{x}}(0), \\
    \tilde{\mathbf{x}}^{(n)}(t) &= \omega_H \int_{0}^{t} e^{A_0(t-t')} A_1(t')\tilde{\mathbf{x}}^{(n-1)}(t')\mathrm{d}t', \;\;\; n \geq 1.
\end{align}
\end{subequations}
However, we are interested only in the first term of $\mathbf{x}(t)$. Since the system is initialized in the code space, $\langle g \rangle_{\tilde{\rho}(0)} = 1$. Then the unperturbed solution is
\[
    \tilde{\mathbf{x}}^{(0)}(t)_1 = \exp\left(-\frac{4\pi \omega}{\kappa}\right) + \mathcal{O}\left(\left(\omega/\kappa\right)^2\right).
\]
Due to the structure of $A(t)$, the first-order correction $\tilde{\mathbf{x}}^{(1)}(t)_1 = 0$. But the second-order correction is
\[
    \tilde{\mathbf{x}}^{(2)}(1)_1 = -\frac{\omega \theta^2}{2\pi\kappa}\exp\left(-\frac{4\pi\omega}{\kappa}\right) + \mathcal{O}\left(\left(\omega/\kappa\right)^2\right),
\]
giving
\[
    \langle g \rangle_{\tilde{\rho}(T)} \approx \bigg(1 - \frac{\omega \theta^2}{2\pi\kappa}\bigg)\exp\bigg(-\frac{4\pi\omega}{\kappa}\bigg).
\]
\end{proof}
\begin{example}
    As a simple example, we again choose the $[\![3, 1, 3]\!]$ bit-flip code, spanned by $\{\ket{000}, \ket{111}\}$ with the stabilizer group $S(0) = \{\sigma_1^z\sigma_2^z, \sigma_1^z\sigma_3^z, \sigma_2^z\sigma_3^z\}$. Suppose the initial state is a pure code state: $\ket{\psi(0)} = \ket{\bar{0}} = \ket{000}$, and we would like to apply the logical $\bar{\sigma^x}$ gate. Choose $H = \otimes_{i=1}^3\sigma_i^x$, and $X = \sigma_1^x\sigma_3^z$, so that it anticommutes with the stabilizers $\sigma_1^z\sigma_2^z$ and $\sigma_1^z\sigma_3^z$, and $H$. By continuously measuring the observables $\{V(t)gV^\dagger(t) \; \forall g \in S(0)\}$, the state should, in principle, rotate into $\ket{111}$ at time $t=2\pi/\omega$. Setting $\kappa=1$, we plot the average fidelity between $\ket{\psi(T)}$ and $\ket{111}$ as a function of the rotation rate $\omega$ in \cref{fig:p_jump}.
\end{example}
\begin{figure}
    \centering
    \includegraphics[width=\linewidth]{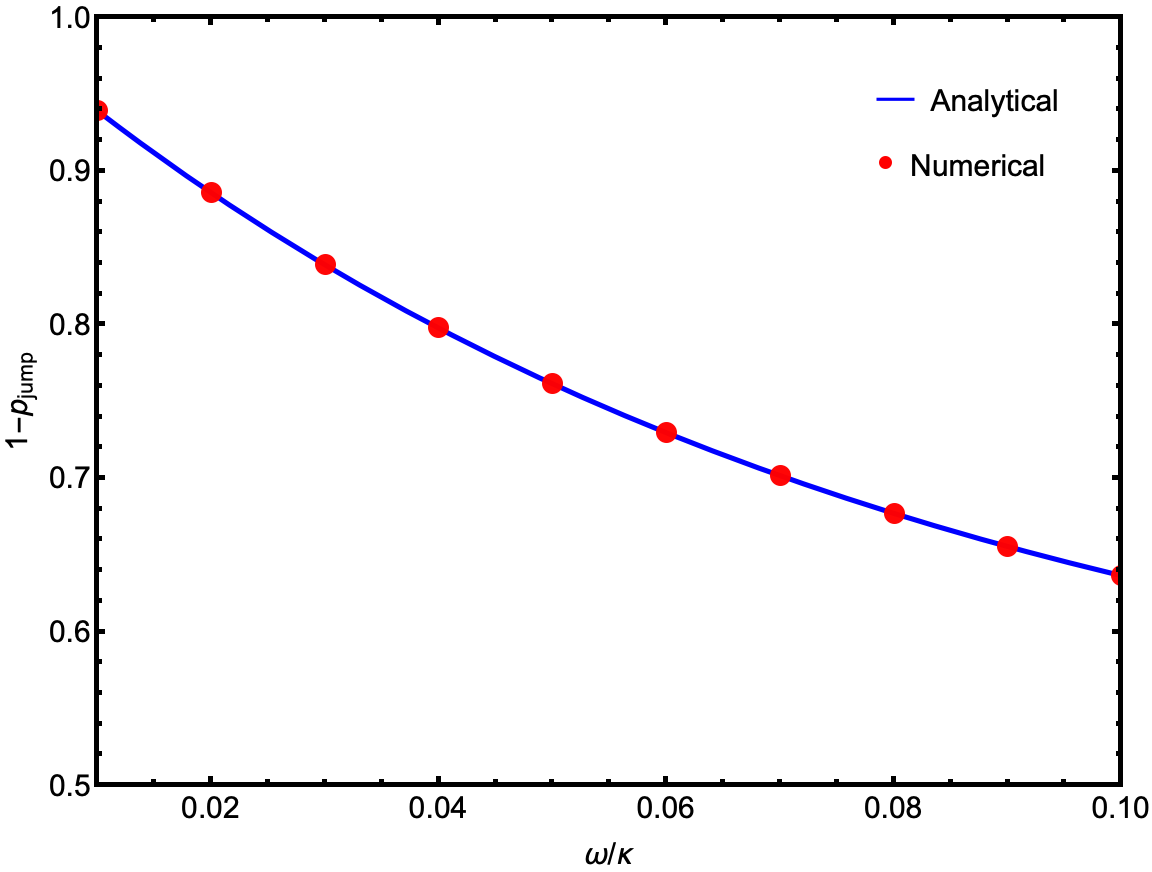}
    \caption{Probability of a jump as a function of $\omega/\kappa$. The analytical equation is \cref{eq:confinement_prob}.}
    \label{fig:p_jump}
\end{figure}

As in the case of projective measurements, we would like this approach to be implicitly fault-tolerant. Whenever the expectation of \emph{any} of the rotated generators is flipped to $-1$ (corresponding to a jump to the error space), the path must be changed on-the-fly from $V(t)$ to $W(t)$. In a real system, however, we do not have access to the expectations themselves: they are filtered from the measurement current \cite{bouten2006introductionquantumfiltering, gough2018introductionquantumfiltering, Lin2020optimalpolynomial, BELAVKIN1992171, chen_brun}. Here, we show one technique to detect such a jump from the measurement current \cite{raghunathan2010continuousmonitoringimprovesinglephoton}.

\subsubsection{Jump detection}
Suppose a set of stabilizer generators $\{g(t)\}$ are continuously measured. The measurement current for one channel (corresponding to a single stabilizer) reads
\[
    i(t) = \mu + \frac{1}{2\sqrt{\kappa}}\frac{\mathrm{d}W_t}{\mathrm{d}t},
\]
where $\mu \equiv \langle g(t) \rangle_{\rho(t)}$. Since $g(t)$ are Pauli operators rotated by a unitary, their eigenvalues are preserved; the two extreme values of $\langle g(t) \rangle$ are $\pm 1$. A jump corresponds to a change of $\mu$ from $+1$ to $-1$. This can be detected using a simple hypothesis test:
\[
\left.\begin{array}{l}
H_+: \mu=+1 \\
H_-: \mu=-1
\end{array}\right\}.
\]
Suppose the output current is averaged over a succession of time-intervals of size $\Delta t$:
\[
    \bar{i}(t) = \frac{1}{\Delta t}\int_t^{t+\Delta t}i(t)\mathrm{d}t.
\]
The averaged current follows a Gaussian distribution with a variable mean $\mu$ and constant variance $1/(4\kappa\Delta t)$. Define $y_k \equiv \bar{i}(t_k)$ to be the discretized output sample. Taking the log-likelihood ratio of all the output samples given the means $\pm 1$, we have
\[
    S_n = \ln\left(\frac{p\big(\{y_1, \cdots, y_n\}|\mu=-1\big)}{p\big(\{y_1, \cdots, y_n\}|\mu=+1\big)}\right).
\]
Since $y_k$ are samples from independent and identical Gaussian distributions,
\[
    S_n = -8\kappa\Delta t\sum_{k=1}^n y_k.
\]
It is easy to see that $S_n$ will have a $V$-shaped profile; whenever the state is in the code space, the true expectations are all $+1$, and $S_n$ will have a negative drift. When the state jumps into the error space, the expectation flips to $-1$, and $S_n$ will then have a positive drift. The instant when there is a change in the direction corresponds to $t_\mathrm{err}$. To detect this change, we define
\[
    m_k \equiv \min_{1\leq j\leq k}S_j.
\]
Then the jump time is
\[
    t_\mathrm{err} = \min_{k}(S_k-m_k)\geq h,
\]
for some small threshold $h$ used to mitigate false positives.

The exact error operator $E(t_\mathrm{err})$, and the updated rotation path $W(t)$ can be found by reparameterizing $\phi$ as $\omega t$ in \cref{eq:error_op} and \cref{eq:new_path_op}, respectively.

\section{Error correcting conditions}\label{sec:qec_conds}

A quantum error-correcting code is designed to correct any error from a \emph{correctable error set} $\mathcal{E}$. During the runtime of this procedure, by construction, the code is changed at every instant. However, we would like for all the instantaneous codes still to be able to correct $\mathcal{E}$. In this section, we derive sufficient conditions on the codes to preserve the correctability of $\mathcal{E}$. Formally, we define a correctable set as follows:

\begin{definition}[Correctable set]
    Given an $[\![n, k, d]\!]$ quantum error-correcting code with stabilizer group $S$, a subset $\mathcal{E}$ of the $n$-qubit Pauli group is defined as a correctable error set if for $E_a, E_b \in \mathcal{E}$, $\mathbb{P}_0 E_b^\dagger E_a \mathbb{P}_0 = \gamma_{ab}\mathbb{P}_0$ $\forall a, b$, where $\gamma_{ab}$ are elements of a Hermitian matrix. (This is the Knill-Laflamme condition \cite{KL_cond}.)
\end{definition}

To preserve correctability, we require that the rotated codes also satisfy the Knill-Laflamme condition with the original correctable error set $\mathcal{E}$:
\[\label{eq:qec_cond}
    \mathbb{P}(t)E_b^\dagger E_a \mathbb{P}(t) = \gamma_{ab}(t)\mathbb{P}(t).
\]
Using the definition for $\mathbb{P}(t)$ from \cref{eq:proj_rotated_code}, this is equivalent to
\[
    \mathbb{P}_0 \big[V^\dagger(t)E_b^\dagger V(t)\big] \big[V^\dagger(t)E_a V(t)\big] \mathbb{P}_0 = \gamma_{ab}(t)\mathbb{P}_0,
\]
i.e., all the rotated errors $\{V^\dagger(t)E_aV(t)\}$ are correctable by the original code. Denote $\mathcal{E}^{(2)}$ as the set of products of all pairs of correctable errors:
\[
    \mathcal{E}^{(2)} = \{E_aE_b| E_a, E_b \in \mathcal{E}\}.
\]
Since we are working with stabilizer codes, all the generators and the error operators are Pauli matrices. Depending on whether an element $D \in \mathcal{E}^{(2)}$ commutes/anticommutes with $H$ and $X$, $V^\dagger(t)D V(t)$ is a linear combination of elements from
\[
    \mathcal{A} =  \{D, HD, XD, HXD\}.
\]
\begin{table}[ht]
    \centering
    \begin{tabular}{c@{}c@{\hskip 1.5em}c@{\hskip 1.5em}c@{\hskip 1.5em}c@{}}
    \toprule
         \textbf{Case} & \textbf{Condition} &  \textbf{Linear combination of} \\
    \midrule
         1 & $[H, D] = [X, D] = 0$ & $D$ \\
         \addlinespace
         2 & $[H, D] = \{X, D\} = 0$ & $D, XD$ \\
         \addlinespace
         3 & $\{H, D\} = [X, D] = 0$ & $D, HD, HXD$ \\
         \addlinespace
         4 & $\{H, D\} = \{X, D\} = 0$ & $D, HD, XD$ \\
         \bottomrule
    \end{tabular}
    \caption{Depending on whether $H$ and $X$ commute or anticommute with the Pauli operator $D$, $V^\dagger(t)D V(t)$ is a linear combination of elements from $\mathcal{A}$ listed in this table.}
    \label{tab:lin_comb_ops}
\end{table}

For correctable errors that commute with both $H$ and $X$, i.e., in case 1 of \cref{tab:lin_comb_ops}, $\mathbb{P}(t)E_b^\dagger E_a \mathbb{P}(t)=\gamma_{ab}(0)\mathbb{P}(t)$ for all $t$ and the error-correcting condition is trivially satisfied through out the loop. To check if the other error operators are correctable, let us first define the syndrome of a Pauli operator $A$, $\vec{s}_A$ as an $(n-k)$-bit binary vector with the $i^\mathrm{th}$ bit being 0 (1) if $A$ commutes (anticommutes) with $g_i$. Then consider the following lemma:
\begin{lemma}\label{lemma:A_sandwiched_in_P0}
    For a Pauli operator $A$, if $\vec{s}_A \neq 0$, then $\mathbb{P}_0 A \mathbb{P}_0 = 0$. If $\vec{s}_A = 0$, then $A$ is either a stabilizer (with $\mathbb{P}_0 A \mathbb{P}_0 = \mathbb{P}_0$), or a logical operator (with $\mathbb{P}_0 A \mathbb{P}_0 \neq \alpha \mathbb{P}_0$).
\end{lemma}
\begin{proof}
See \cref{sec:A_sandwiched_in_P0_proof}.
\end{proof}

The rotated error operators are correctable, in principle, as long as they satisfy the KL-condition for the rotated codes. For simplicity, we choose the matrix $\gamma$ to be diagonal. This provides \emph{sufficient} conditions on the choice of $X$. We find that the operator $HD$ always satisfies the error-correcting condition:
\begin{proposition}\label{prop:HEbEa_correctable}
    For a logical operator $H$ and $D \in \mathcal{E}^{(2)}$, $\mathbb{P}_0 HD \mathbb{P}_0 \propto \mathbb{P}_0$.
\end{proposition}
\begin{proof}
    Since $H$ is a logical operator, $\vec{s}_H=0$. Moreover, $\mathbb{P}_0 HD \mathbb{P}_0 =0$ unless $\vec{s}_{D} = 0$. Observe that the product of any pair of correctable errors has weight at most $d-1$:
    \[
    \mathrm{wt}(D) \leq d-1 \quad \forall D \in \mathcal{E}^{(2)}.
    \]
    However, the minimum weight of any logical operator is $d$. Therefore, $D$ cannot be a logical operator, making its syndrome 0 only if $D$ is a stabilizer. Since all stabilizers commute with $H$, we are in case 1 or case 2 of \cref{tab:lin_comb_ops}; the operators we are concerned with are $D$ and $XD$. If $D$ is a stabilizer, 
    \[
    \vec{s}_{XD} = \vec{s}_X \neq 0.
    \]
    By \cref{lemma:A_sandwiched_in_P0}, $\mathbb{P}_0 XD \mathbb{P}_0 = 0$, satisfying the error-correcting conditions.
\end{proof}
For any logical operator $H$, the operator $HD \;\forall D \in \mathcal{E}^{(2)}$ has a trivial influence on correctability, i.e., no extra requirements on the code or $X$ are needed to satisfy the error-correcting condition. However, as we shall see, the operators $HXD$ and $XD$ are nontrivial and they impose certain restrictions on $X$.

\begin{theorem}\label{thm:qec_conditions}
    For a stabilizer code with correctable error set $\mathcal{E}$, a set of sufficient conditions such that all the rotated codes can correct $\mathcal{E}$ is:
    \begin{itemize}
        \item Hamming weight of $X> d-1$.
        \item All stabilizers within Hamming distance $d-1$ ($d-2$ for even $d$) of $X$ anticommute with $X$.
        \item All logical operators within Hamming distance $d-1$ ($d-2$ for even $d$) of $X$ commute with $X$ and anticommute with $H$.
    \end{itemize}
\end{theorem}
\begin{proof}
$D$ trivially satisfies the error correcting condition and $HD$ satisfies it by \cref{prop:HEbEa_correctable}. The only other operators we are concerned with are $XD$ and $HXD$. Observe that $\mathbb{P}_0HXD\mathbb{P}_0 = \mathbb{P}_0XD\mathbb{P}_0 = 0$ unless $\vec{s}_{X} = \vec{s}_{D}$. That happens when $X=D$, $XD$ is a stabilizer, or when $XD$ is a logical operator. Let us consider each case separately:

$X=D$. Since $X$ commutes with itself and anticommutes with $H$, we are in case 3 of \cref{tab:lin_comb_ops}, and the nontrivial operators are $HD$ and $HXD$. By \cref{prop:HEbEa_correctable}, $HD$ satisfies the condition. However, $\mathbb{P}_0 HXD \mathbb{P}_0 = \mathbb{P}_0 H \mathbb{P}_0 \; \not\propto \; \mathbb{P}_0$, and it is not correctable. The only way to avoid this scenario is to choose $X$ such that its Hamming weight is greater than $d-1$ so that it can never equal $D$.

$XD$ is a stabilizer. Since every stabilizer commutes with every logical operator, $[H, XD] = 0$. Since $\{H, X\}=0$, it must be true that $\{H, D\} = 0$ (case 3 or case 4). Again, $HD$ satisfies the condition. In case 3, $\mathbb{P}_0HXD\mathbb{P}_0 = \mathbb{P}_0 H \mathbb{P}_0 \; \not\propto \; \mathbb{P}_0$ and the error-correcting condition is not satisfied. However, we {\it can} satisfy the condition in case 4, i.e., $\{X, D\} = 0$. Equivalently, we need to satisfy $\{X, XD\} = 0$, i.e., all stabilizers within Hamming distance $d-1$ ($d-2$ for even $d$) of $X$ should anticommute with $X$.

$XD$ is a logical operator. We can immediately see that $\mathbb{P}_0 XD\mathbb{P}_0 \; \not\propto \; \mathbb{P}_0$ and $\mathbb{P}_0 HXD\mathbb{P}_0 \; \not\propto \; \mathbb{P}_0$. So, the only way to satisfy the error-correcting condition is in case 1, i.e., $[H, D] = [X, D] = 0$. Equivalently, $\{H, XD\} = [X, XD] = 0$, i.e., all logical operators within Hamming distance $d-1$ ($d-2$ for even $d$) of $X$ commute with $X$ and anticommute with $H$.
\end{proof}

A simple numerical search demonstrates that the conditions in \cref{thm:qec_conditions} can be satisfied for the Shor $[\![9,1,3]\!]$ code, but not for the perfect $[\![5,1,3]\!]$ code or the Steane $[\![7,1,3]\!]$ code. However, we can augment the codes with extra ancillas so that the conditions can be satisfied, as we shortly show.

\subsection{Deficient syndromes}

Consider the $[\![9,1,3]\!]$ Shor code with a logical operator $H = \sigma_1^z\sigma_4^z\sigma_7^z$, and $X = \sigma_1^x\sigma_4^x\sigma_7^x$, so that $X$ anticommutes with $H$, and with the stabilizer generators $\sigma_1^z\sigma_2^z$, $\sigma_4^z\sigma_5^z$ and $\sigma_7^z\sigma_8^z$. A numerical search shows that $\vec{s}_{X}\neq \vec{s}_{D} \; \forall D \in \mathcal{E}^{(2)}$. With this choice of $H$, $X$ and $\mathcal{E}$, we can satisfy the error-correcting conditions throughout the loop.

It is easy to verify that the syndromes of all pairs of errors include all syndromes for the perfect $[\![5, 1, 3]\!]$ and the Steane $[\![7, 1, 3]\!]$ codes; for any $X$, there will be choices of $E_bE_a$ with the same syndrome, and no $X$ satisfies all the required commutation relations. For this construction, it is helpful for a code to have more ``slack'' to satisfy the conditions. We can augment the codes by appending ancilla qubits, thereby increasing the number of syndromes. Consider the Steane $[\![7,1,3]\!]$ code for example; we add to this code one ancilla qubit with stabilizer generator $Z$, i.e., the extra qubit is always in the state $\ket{0}$. Then the projector on to the code space becomes 
\[
    \mathbb{\bar{P}}_0 \equiv \mathbb{P}_0\otimes \ket{0}_A\bra{0} ,
\]
where the subscript $A$ stands for the ancilla. The bar stands for the updated code. We choose $\bar{H} = H\otimes I_A$ and $\bar{X}=X\otimes \sigma^x_A$. For instance, we can choose $H=Z^{\otimes7}$ and $X=XZIIIII$. Assume that the correctable set contains all weight-$1$ Pauli operators on the eight qubits plus the identity. (Note that this is fewer errors than this code could correct: it could correct any weight-$1$ error on the seven qubits, together with any error on the ancilla.) We claim that this construction satisfies \cref{thm:qec_conditions} based on two observations:
\begin{itemize}
    \item $X$ has a syndrome different from any weight-$1$ error. So $\bar{X}=X\otimes \sigma^x_A$ has a syndrome different from any $E_bE_a$ and $\mathbb{\bar{P}}_0 \bar{X} E_bE_a \mathbb{\bar{P}}_0 = \mathbb{\bar{P}}_0 H \bar{X} E_bE_a \mathbb{\bar{P}}_0 = \mathbb{\bar{P}}_0 HE_bE_a \mathbb{\bar{P}}_0 = 0$ for any $E_b, E_a \in \mathcal{E}$.
    \item Since the code is non-degenerate, $\mathbb{\bar{P}}_0E_bE_a\mathbb{\bar{P}}_0=0$ for all $b\neq a$. For $b=a$, $E_aE_a=I$, which commutes with $H$ and $X$.
\end{itemize}

Note that this construction fails for the $[\![5,1,3]\!]$ code; since it is ``perfect'', there is a weight-$1$ error corresponding to every syndrome. So any choice $X$ is distance $1$ away from at least one stabilizer or logical operator---for any $X\otimes X_A$ there is a pair $E_bE_a$ with the same syndrome. However, this code can be made to work by using {\it two} ancilla qubits $A_1$ and $A_2$, with any $H=H\otimes I_{A_1}\otimes I_{A_2}$. Choose $H=\otimes_{i=1}^5\sigma_i^z$, and $X=\sigma^x_1\otimes \sigma^x_{A_1}\otimes \sigma^x_{A_2}$. Then no $E_bE_a$ will have the same syndrome as $X$ and the code is nondegenerate. Here, we transformed the $[\![5,1,3]\!]$ code into a $[\![7,1,3]\!]$ code by adding two extra qubits and two extra stabilizers $\sigma_6^z$ and $\sigma_7^z$ to the original 4 stabilizers. By increasing the number of syndromes, we ensure that all the rotated codes can still correct the natural correctable set of the original code.

In \cref{sec:fault_tolerance}, we showed a way to correct a measurement-induced error, which appears as a Markovian $X$ error. However, in the error-correcting conditions considered here, the correctable set of the original code need not include the $X$ operator. Therefore, a measurement-induced error can become uncorrectable if the original correctable set contains a stabilizer element that anticommutes with $X$. To see this, observe the action of a correctable error $E$ on an instantaneous code state:
\[
    \ket{\psi_E(t)} = E\ket{\bar\psi(t)} 
    = U(t)E(t)\ket{\bar\psi},
\]
where we define the rotated error operator
\[
E(t)=U^\dagger(t) E U(t)
\]
acting on the unrotated code state, and \[U(t)=V(t)\exp\left(-i\frac{\theta}{4\pi}\sin(2\omega t)H\right).\] If $E$ is a stabilizer that anti-commutes with $X$, the \emph{unrotated} error state $E(t)\ket{\bar\psi}$ becomes a superposition of states with syndromes $\textbf{s}_E=\vec0$ and $\textbf{s}_{EX}=\textbf{s}_X$. Since we continue measuring even after an error has happened, the state relaxes onto one of the syndrome spaces. In the latter case, the state acquires a logical phase that may be different from that of a measurement-induced error. Therefore, the correction for $X$ differs from that of $EX$ by a logical rotation.

A straightforward method to avoid this scenario is to choose an ancilla code such that its stabilizers are not included in the natural correctable set. For instance, we can append $2$ ancilla qubits initialized in the Bell-state to the $[\![5, 1, 3]\!]$ code. The Bell-state has stabilizers generated by $\langle X_{A_1}X_{A_2}, Z_{A_1}Z_{A_2}\rangle$, where $A_1$ and $A_2$ denote the ancilla qubit indices. This code can correct any arbitrary single-qubit Pauli error on the $7$ qubits, as well as the $X$ error.

\section{Discussion}\label{sec:discussion}

In this paper, we have proposed a framework to generate continuous-path holonomies in arbitrary stabilizer codes by continuously measuring the rotated stabilizers. When the rotation forms a loop, the holonomy corresponds to a logical rotation by an arbitrary angle in the code space. Due to the inherent stochasticity of quantum measurements, the code state may jump into the rotated error space if the rotation is not made adiabatically. We showed a method to change the path so that at the end of the evolution, we still achieve the desired holonomy; this provides flexibility in choosing the gate times. In this work, we have not treated the effects of external noise. In practice, however, the system constantly interacts with its environment, causing either Markovian or non-Markovian errors. Since we continuously measure the stabilizers, the Zeno effect can suppress the non-Markovian noise \cite{hsu_brun, convy_machine_2022, ahn_continuous_2002, lanka2025optimizingcontinuoustimequantumerror, nila2025continuousquantumcorrectionmarkovian, qze_encoded_comp, Dominy_2013}. To eliminate Markovian errors, however, we would have to perform error-correction. We would like to find an analogous way to change the path that depends both on the correctable error and when it occurred so that the logical gate is successfully performed. This ensures full fault-tolerance of the protocol. We also gave sufficient conditions on the measurement observables so that the Knill-Laflamme condition is satisfied for all the rotated codes. When these conditions are not satisfied, we can still use this protocol by adding ancilla qubits to the original code, choosing the ancilla code such that its stabilizers differ from all the errors in the correctable set. We analyze these questions and provide constructive solutions in the companion paper \cite{lanka2026steeringpathsmidflightfaulttolerance}. Nevertheless, a more general construction that includes the $X$ error as one of the correctable errors is desired.

The measurement procedure described in this work may be physically realized by engineering a time-dependent Hamiltonian between the data qubits and an ancilla qubit, while measuring the ancilla qubit continuously. For smaller codes, it may be possible to engineer this Hamiltonian interaction such that it is always pairwise local \cite{marvian_engr_hams, chen_brun, Hung_2016, Haas_2019}. In this work, we have assumed stabilizer codes with Pauli correctable sets for simplicity. In principle, we can generalize this to arbitrary codes. For instance, we can consider codes defined on \emph{qudit}-systems but encoding logical qubits. Such systems may be better suited to engineering such Hamiltonian interactions. Many open questions remain about this approach and how to implement it in realistic systems, but it is a remarkably interesting and fertile topic for future study.

\acknowledgements

We thank Daniel Lidar, Ben Reichardt, Eli Levenson-Falk, Itay Hen and Onkar V. Apte for useful discussions. This work was supported in part by the U. S. Army Research Laboratory and the U. S. Army Research Office under contract/grant number W911NF2310255, and by NSF Grant FET-2316713.

\bibliography{references}

\appendix

\begin{widetext}

\section{Proof of \cref{lemma:small_rot_sandwiched_P0}}\label{app:small_rot_sandwiched_P0_proof}
Assume that the rotation increment $\delta\phi\ll1$. Then the rotation operator $V^\dagger(\varphi)V(\varphi-\delta\phi)$ is \emph{weak}; sandwiching it between $\mathbb{P}_0$, we obtain
\[
    \mathbb{P}_0V^\dagger(\varphi)V(\varphi-\delta\phi)\mathbb{P}_0 = \mathbb{P}_0 e^{-i\varphi X}e^{-i\frac{\theta}{2\pi}\varphi H}e^{i\frac{\theta}{2\pi}(\varphi-\delta\phi) H}e^{i(\varphi-\delta\phi) X}\mathbb{P}_0.
\]
Using that fact that $\mathbb{P}_0X\mathbb{P}_0 = \mathbb{P}_0XH\mathbb{P}_0 = 0$, we obtain
\begin{subequations}
    \begin{align}
        \mathbb{P}_0 V^\dagger(\varphi) V(\varphi-\delta\phi) \mathbb{P}_0 &=
            \cos\left(\frac{\theta}{2\pi} \delta\phi\right) \cos(\delta\phi) \mathbb{P}_0 - i\sin\left(\frac{\theta}{2\pi} \delta\phi\right) \cos(2\varphi - \delta\phi) H \mathbb{P}_0 \\
        &\quad \equiv c_\varphi e^{-i\xi_\varphi H} \mathbb{P}_0.
    \end{align}
\end{subequations}
The normalizing constant $c_\varphi$ is given by
\[
\begin{split}
    c_\varphi 
    &= \left[
        \cos^2\left(\frac{\theta}{2\pi} \delta\phi\right) \cos^2(\delta\phi) + \sin^2\left(\frac{\theta}{2\pi} \delta\phi\right) \cos^2(2\varphi - \delta\phi) \right]^{1/2} \\
    &= \left[
        1 - \left(\frac{\theta}{2\pi} \delta\phi\right)^2 (1-\delta\phi)^2 + \left(\frac{\theta}{2\pi} \delta\phi\right)^2 \cos^2(2\varphi - \delta\phi) \right]^{1/2} + \mathcal{O}(\delta\phi^4) \\
    &= \left[
        1 - \delta\phi^2\left(1 + \frac{\theta^2}{4\pi^2}\sin^2(2\varphi-\delta\phi)\right) \right]^{1/2} + \mathcal{O}(\delta\phi^4) \\
    &= \left[
        1 - \delta\phi^2\left(1 + \frac{\theta^2}{4\pi^2}\sin^2(2\varphi)\right) \right]^{1/2} + \mathcal{O}(\delta\phi^3) \\
    &= 1 - \frac{\delta\phi^2}{2} \left( 1 + \frac{\theta^2}{4\pi^2} \sin^2(2\varphi) \right) + \mathcal{O}(\delta\phi^3),
\end{split}
\]
and the angle $\xi_\varphi$ follows:
\[
    \tan(\xi_\varphi) = \frac{\sin\left(\frac{\theta}{2\pi}\delta\phi\right)\cos(2\varphi-\delta\phi)}{\cos\left(\frac{\theta}{2\pi}\delta\phi\right)\cos(\delta\phi)} 
    = \frac{\theta}{2\pi}\cos(2\varphi)\delta\phi + \mathcal{O}(\delta\phi^2),
\]
where we used the fact that $\xi_\varphi = \arctan(\xi_\varphi) + \mathcal{O}(\delta\phi^3)$.

\section{Proof of \cref{lemma:A_sandwiched_in_P0}}\label{sec:A_sandwiched_in_P0_proof}

When $\vec{s}_A \neq 0$, then by definition, $A$ anticommutes with at least one of the stabilizer generators. Let $C$ ($AC$) be the set of generators that commute (anticommute) with $A$. Then,
\[
    \begin{aligned}
        \mathbb{P}_0 A\mathbb{P}_0 &= \left(\prod_{g\in S(0)}\frac{I+g}{2}\right)A\left(\prod_{g\in S(0)}\frac{I+g}{2}\right) \\
        &= \left(\prod_{g\in S(0)}\frac{I+g}{2}\right)\left(\prod_{g \in S_\mathrm{c}(0)}\frac{I-g}{2}\right)\left(\prod_{g \in S_\mathrm{c}(0)}\frac{I+g}{2}\right) A
        = 0.
    \end{aligned}
\]
On the other hand, when $\vec{s}_A = 0$, $A$ commutes with all the generators. Then $A$ is an element of the normalizer of the stabilizer group:
\[
    A \in N(S) = \{g \in \mathcal{P}_n : g S_ig^\dagger \in S \; \forall S_i \in S \},
\]
which includes both stabilizers and logical operators.

\section{Instantaneous code state}\label{sec:inst_estate}

For simplicity, we assume that the evolution happens with projective measurements of the rotated code space, defined by the rotation operator \cref{eq:rotation_unitary}. But the obtained error operator should be the same in the case of continuous measurements under the limit of small rotations.

When the state evolves as the $+1$-eigenstate of the rotated code space projector (i.e., no error), the resultant state after rotation by $\zeta$ is
\[
\begin{split}
    \ket{\psi_\zeta} &= \mathbb{P}(\zeta)\cdot \mathbb{P}(\zeta-\delta\phi)\cdots\mathbb{P}_0\ket{\bar{\psi}}\\
    &= [V(\zeta)\mathbb{P}_0 V^\dagger(\zeta)]\cdots[V(\delta\phi)\mathbb{P}_0V^\dagger(\delta\phi)]\ket{\bar{\psi}}.
\end{split}
\]
This is a product of factors of the form $\mathbb{P}_0 V^\dagger(\varphi)V(\varphi-\delta\phi)\mathbb{P}_0$, where $\varphi=0...\zeta$. From \cref{lemma:small_rot_sandwiched_P0}, we obtain the ``partial'' $H$ rotation angle $\sum_l\xi_{l\delta\phi}$. Observe that all the factors depend only on $H$ and $\mathbb{P}_0$, which commute with each other. So the state after rotation by $\zeta$, conditioned on all successful measurements, is
\[
\begin{split}
    \ket{\psi_\zeta} &= V(\zeta)\exp\left(-i \sum_{l=1}^{\ceil{\zeta/\delta\phi}}\xi_{l\delta\phi}H\right)\ket{\bar{\psi}} \\
    &= V(\zeta) \exp\left(-i\frac{\theta}{4\pi}\sin(2\zeta)H\right)\ket{\bar{\psi}} + \mathcal{O}(\delta\phi^2).
\end{split}
\]
Re-parameterizing the angle $\zeta = \omega t$ as a function of continuous time, we obtain the instantaneous code state
\[
    \ket{\psi(t)} = V(t) \exp\left(-i\frac{\theta}{4\pi}\sin(2\omega t)H\right)\ket{\bar{\psi}}.
\]

\end{widetext}

\end{document}